\documentclass[a4paper, 11pt]{amsart}
\usepackage{amssymb,array}
\usepackage{url}
\usepackage{graphicx} 
\usepackage{subfigure}

\usepackage[left=1in, right=1in, top=1in, bottom=1in]{geometry}

\renewcommand{\leq}{\leqslant}
\renewcommand{\geq}{\geqslant}

\newcommand{\N}{\mathbb{N}}

\newcommand{\C}{\mathbb{C}}
\newcommand{\E}{\mathbb{E}}
\renewcommand{\P}{\mathbb{P}}

\newcommand{\U}{\mathcal{U}}
\newcommand{\D}{\mathcal{D}}

\renewcommand{\S}{\mathcal{S}}

\newcommand{\ol}{\overline}

\newcommand{\isom}{\simeq}

\DeclareMathOperator{\trace}{Tr}

\DeclareMathOperator{\id}{id}
\DeclareMathOperator{\Wg}{Wg}
\DeclareMathOperator{\Mob}{Mob}
\DeclareMathOperator{\Cat}{Cat}

\DeclareMathOperator*{\Ex}{\mathbb{E}}

\renewcommand{\phi}{\varphi}
\renewcommand{\epsilon}{\varepsilon}
\newcommand{\iy}{\infty}

\newtheorem{theorem}{Theorem}[section]
\newtheorem{definition}[theorem]{Definition}
\newtheorem{proposition}[theorem]{Proposition}

\newtheorem{lemma}[theorem]{Lemma}

\newtheorem{corollary}[theorem]{Corollary}

\def\be{\begin{eqnarray}}
\def\ee{\end{eqnarray}}
\def\bee{\begin{equation}}
\def\eee{\end{equation}} 

\begin{document}

\title{Low entropy output states for products of random unitary channels}
\author{Beno\^{\i}t Collins}
\address{
D\'epartement de Math\'ematique et Statistique, Universit\'e d'Ottawa,
585 King Edward, Ottawa, ON, K1N6N5 Canada
and 
CNRS, Institut Camille Jordan Universit\'e  Lyon 1, 43 Bd du 11 Novembre 1918, 69622 Villeurbanne, 
France} 
\email{bcollins@uottawa.ca}
\author{Motohisa Fukuda}
\address{Zentrum Mathematik, M5,
Technische Universit\"{a}t M\"{u}nchen,
Boltzmannstrasse 3,
85748 Garching, Germany}
\email{m.fukuda@tum.de}
\author{Ion Nechita}
\address{CNRS, Laboratoire de Physique Th\'eorique, IRSAMC, Universit\'e de Toulouse, UPS, 31062 Toulouse, France and Zentrum Mathematik, M5,
Technische Universit\"{a}t M\"{u}nchen, Boltzmannstrasse 3,
85748 Garching, Germany}
\email{nechita@irsamc.ups-tlse.fr}

\begin{abstract}
In this paper, we study the behaviour of the output of pure entangled states after being transformed by a product of conjugate random unitary channels. This study is motivated by the counterexamples by Hastings \cite{hastings} and Hayden-Winter \cite{hayden-winter} to the additivity problems. 
In particular, we study in depth the difference of behaviour between random unitary channels and generic random channels. In the case where the number of unitary operators is fixed, we compute the limiting eigenvalues of the output states. In the case where the number of unitary operators grows linearly with the dimension of the input space, we show that the eigenvalue distribution converges to a limiting shape that we characterize with free probability tools.
In order to perform the required computations, we need a systematic way of dealing with moment problems for random matrices whose blocks are i.i.d.\ Haar distributed unitary operators. This is achieved by extending the graphical Weingarten calculus introduced in \cite{cn1}.
\end{abstract}

\maketitle

\section{Introduction}

In Quantum Information Theory, random unitary channels are completely positive, 
trace preserving and unit preserving maps that can be written as
$$\Phi(\rho) = \sum_{i=1}^kw_i U_i \rho U_i^* ,$$
where $U_i$ are unitary operators acting on $\mathbb C^n$ and $w_i$ are positive weights that sum up to one.

This class of quantum channels is very important, not only because the action of such a channel can be understood as randomly applying one of the unitary transformations $U_i$ with respective probabilities $w_i$, but also because it has highly non-classical properties. For example, random unitary channels have been used to disprove the additivity of minimum output entropy \cite{hastings}. 

In this paper, we are interested in a natural setting in which we take a convex combination of $k$ \emph{random} unitary evolutions; in other words, we choose the unitary operators $U_i$ at random and independently from the unitary group, following the Haar distribution. The behaviour of this kind of quantum channels has been extensively studied in the literature \cite{aubrun, hlsw, hastings}.

We are principally interested in the study of the moments of typical outputs for products of conjugated random unitary channels. One of the main results of our
paper is that the typical outputs are deterministic when one takes a product of such a channel and its complex conjugate and applies it to entangled input states, in the spirit of \cite{cn1,cn3,cfn1}. 
More precisely, our main results can be stated informally as follows (for precise statements, see Theorems \ref{thm:bell-phenomenon} and \ref{thm:k-sim-n}):
\begin{theorem}
Consider the output state $Z_n = [\Phi \otimes \bar \Phi](\psi_n \psi_n^*)$, given as the image of a ``well behaved'' pure state $\psi_n$ 
under
the product of conjugate random unitary channels.

If $k$, the number of unitary operators $U_i$, is fixed and $n \to \infty$, then the set of the $k^2$ non-zero eigenvalues of $Z_n$ is 
$$\{w_iw_j \, : \, i,j = 1, \ldots, k, i \neq j\} \cup \{s_i \, : \, i=1, \ldots, k\},$$
where the numbers $s_i$ depend on $w$ and some parameter $m$ quantifying the overlap between the input state $\psi_n$ and the Bell state $\phi_n$.

If both dimensions $k$ and $n$ grow to infinity, but at a constant ratio $k/n \to c$, then the largest eigenvalue of $Z_n$ behaves as $\mathrm{Const.}/k$ and the rest of the spectrum, when properly rescaled, converges towards a compound free Poisson distribution that can be characterized in terms of the behaviour of the ratio $c$ and the weights $w_i$. 
\end{theorem}

We strengthen previous results on the eigenvalues of the output states $Z_n$, obtained by linear algebra techniques, by computing exactly the asymptotic spectrum of the output states. Such precise results, along with results for single channels \cite{cfn2}, may improve existing violations of the additivity of the minimum output entropy or similar quantities. For the purpose of studying these random channels, we extend the graphical calculus introduced in \cite{cn1} to the case of multiple independent, identically distributed ( i.i.d.\ ) random Haar unitary matrices, see Theorem \ref{thm:V-graphical-calculus}.

Our paper is organised as follows. 
In Section \ref{sec:random-unitary-channels} we recall facts about quantum channels and give the definition of random unitary channels. In Section \ref{sec:graphical-calculus},
 we extend Weingarten calculus to include the case of several independent Haar unitary matrices. Section \ref{sec:measures-isometry} provides a comparison of Haar random isometries and block random isometries. Even though it is intuitively clear that the two are quite different from each other, we give explicit instances of the differences via a few explicit moment computations. These moment computations serve also as a warm-up for Sections \ref{sec:fixed-k} and \ref{sec:k-sim-n}, where the behaviour of output of entangled pure states is considered under different scalings; these sections contain the main results of the paper. In Section \ref{sec:fixed-k} we treat the case of a fixed number of unitary operators, while in Section \ref{sec:k-sim-n} we consider the scaling $k/n \to c$. We analyze thoroughly the entropies of the limiting objects, finding the parameter values which yield the output states with the least entropy. Finally, we make a few concluding remarks in Section \ref{sec:conclusions}.

\section{Random unitary channels}\label{sec:random-unitary-channels}

\subsection{Our model}
Random unitary channels are natural models among quantum channels and have
played a key role in the research around violations of different quantities related to classical capacities of quantum channels.

Random unitary channels are unit-preserving quantum channels that can be written as 
\be\label{eq:def-Phi}
\Phi(\rho) = \sum_{i=1}^k w_i U_i \rho U_i^* 
\ee
where $U_i \in \mathcal{U} (n)$ are independent random unitary operators with \emph{uniform} (or \emph{Haar}) distribution and $w_i$ are fixed positive weights that sum up to one. Hastings used a similar model to prove additivity violation of minimal output entropy \cite{hastings}.
In our model, a unitary operator $U_i$ is applied to 
an
input with some fixed probability $w_i$, 
while in Hastings' model those probabilities were chosen randomly. 

\subsection{Stinespring picture and complementary channel} 
The above random channel $\Phi:M_n(\mathbb C) \to M_n(\mathbb C)$ is described in the Stinespring picture as
\be
\Phi (\rho) = \trace_{k} \left[\tilde V \rho \tilde V^* \right]
\ee 
where
\be\label{isometry}
\tilde V = [\mathrm{diag}(\sqrt{w_1},\ldots,\sqrt{w_k}) \otimes \mathrm{I}_n] \, V  = \begin{pmatrix}
\sqrt{w_1} \, U_1 \\ \vdots \\ \sqrt{w_k}\,  U_k 
\end{pmatrix} 
: \C^n \rightarrow \C^k \otimes \C^n  
\ee  
is an isometry, 
$$V = \sum_{i=1}^k e_i \otimes U_i$$
is a stack of unitary operators and $\trace_{k}$ is the partial trace operator $\trace \otimes\mathrm{id}$ acting on $M_k(\mathbb C) \otimes M_n(\mathbb C)$.
In matrix representation, 
\be
\tilde V \rho \tilde V^*  = 
\begin{pmatrix}
w_1 U_1 \rho U_1^* & \ldots & \sqrt{w_1 w_k} U_1 \rho U_k^* \\
\vdots & \ddots & \vdots \\
\sqrt{w_k w_1} U_k \rho U_1^* & \ldots & w_k U_k \rho U_k^* 
\end{pmatrix} 
\ee
As one sees, 
tracing out the first space (the environment space $\C^k$) gives the sum of diagonal blocks,
which is the quantum channel $\Phi$. 
On the other hand, 
if we switch the roles of output and environment spaces,
i.e., if we trace out the second space (the output space $\C^n$),
we get the complementary channel $\Phi^C: M_n(\mathbb C) \to M_k(\mathbb C)$
\cite{Holevo05,KMNR}:
\be
\left( \Phi^C (\rho) \right)_{i,j}= \sqrt{w_iw_j}\trace \left[U_i \rho U_j^* \right]
\ee 
A quantum channel and its complementary channel share 
the same output eigenvalues, up to zeroes, for any pure input via Schmidt decomposition, 
see  \cite{Holevo05,KMNR} for details. 

Recall that the Shannon entropy functional, defined for probability vectors $p=(p_i)$
$$H(p) = -\sum_i p_i \log p_i,$$
extends to unit trace, positive matrices (or quantum states) $\rho$ via functional calculus
$$H(\rho) = -\trace(\rho \log \rho).$$

For a quantum channel $\Phi$, one can define its \emph{minimal output entropy}
$$H_\text{min}(\Phi) = \min_{\rho \text{ state}} H(\Phi(\rho)).$$

By using convexity properties, one can show that the minimum above is reached on pure states (i.e. rank one projectors) so that the value of the minimum output entropy does not change when by replacing $\Phi$ by $\Phi^C$. Our interest in complementary channels is motivated by the fact that, often, the size of the environment is smaller than the output size, so output states are easier to study for $\Phi^C$. 

In the seminal paper of Hastings \cite{hastings}, violation of additivity was found when the dimensions of input and output spaces are much larger than that of the environment space of the channel. In this paper, we shall study two asymptotic regimes, $k$ fixed and $n \to \infty$ and then $k,n \to \infty$, with a linear growth $k/n \to c$.

\section{Weingarten calculus for several independent unitary matrices}\label{sec:graphical-calculus}

The method of graphical calculus was introduced in \cite{cn1} and later used in \cite{cn3, cn-entropy, cnz, cfn1} to study random matrix models in which Haar distributed unitary operators played a major role. In particular, the limiting eigenvalue distribution of
$[\Phi \otimes \bar \Phi] (\psi_n\psi_n^*)$,
where $\Phi$ is random quantum chanel and $ \psi_n $ is a generalised Bell state,
was calculated in \cite{cfn1}. The aim of this paper is to perform similar computations for random unitary channels instead of random channels.

In the previous work in which random quantum channels were analysed using the graphical calculus method \cite{cn1,cn3,cn-entropy,cfn1}, the isometry defining the channel in the Stinespring picture was a truncation of a Haar-distributed random unitary matrix. Hence, only one random unitary operator was needed to perform the computations. In this paper, we need to develop a new technique of graphical calculus for $k$ i.i.d. Haar-distributed unitary matrices. 

We start with the usual, 1-matrix, graphical Weingarten calculus and then
generalise it to cover the case of block-isometries built up from independent Haar-distributed unitary matrices. 

\subsection{Weingarten formula for a single matrix} 
 Let us start by recalling the definition of a combinatorial object of interest, the unitary Weingarten function.
\begin{definition}
The unitary Weingarten function 
$\Wg(n,\sigma)$
is a function of a dimension parameter $n$ and of a permutation $\sigma$
in the symmetric group $\S_p$. 
It is the inverse of the function $\sigma \mapsto n^{\#  \sigma}$ under the convolution for the symmetric group ($\# \sigma$ denotes the number of cycles of the permutation $\sigma$).
\end{definition}

Notice that the  function $\sigma \mapsto n^{\# \sigma}$ is invertible when $n$ is large, as it
behaves like $n^p\delta_e$
as $n\to\infty$.
Actually, if $n<p$ the function is not invertible any more, but we can
keep this definition by taking the pseudo inverse 
and the theorems below will still hold true
(we refer to \cite{collins-sniady} for historical references and further details). We shall use the shorthand notation $\Wg(\sigma) = \Wg(n, \sigma)$ when the dimension parameter $n$ is clear from context.

The function $\Wg$  is used to compute integrals with respect to 
the Haar measure on the unitary group (we shall denote by $\U(n)$ the unitary group acting on an $n$-dimensional Hilbert space). The first theorem is as follows:

\begin{theorem}\label{thm:Wg}
 Let $n$ be a positive integer and
$i=(i_1,\ldots ,i_p)$, $i'=(i'_1,\ldots ,i'_p)$,
$j=(j_1,\ldots ,j_p)$, $j'=(j'_1,\ldots ,j'_p)$
be $p$-tuples of positive integers from $[n]= \{1, 2, \ldots, n\}$. Then
\begin{multline}\label{eq:Wg}
\int_{\U(n)}U_{i_1j_1} \cdots U_{i_pj_p}
\overline{U_{i'_1j'_1}} \cdots
\overline{U_{i'_pj'_p}}\ dU=\\
\sum_{\sigma, \tau\in \S_{p}}\delta_{i_1i'_{\sigma (1)}}\ldots
\delta_{i_p i'_{\sigma (p)}}\delta_{j_1j'_{\tau (1)}}\ldots
\delta_{j_p j'_{\tau (p)}} \Wg (n,\tau\sigma^{-1}).
\end{multline}

If $p\neq p'$ then
\begin{equation} \label{eq:Wg_diff} \int_{\U(n)}U_{i_{1}j_{1}} \cdots
U_{i_{p}j_{p}} \overline{U_{i'_{1}j'_{1}}} \cdots
\overline{U_{i'_{p'}j'_{p'}}}\ dU= 0.
\end{equation}
\end{theorem}

Since we shall perform integration over \emph{large} unitary groups, we are interested in the values of the Weingarten function in the limit $n \to \iy$. The following result encloses all the data we need for our computations
about the asymptotics of the $\Wg$ function; see \cite{collins-imrn} for a proof.

\begin{theorem}\label{thm:mob} For a permutation $\sigma \in \S_p$, we write $c \in \sigma$ to denote the fact that $c$ is a cycle of $\sigma$. Then
\begin{equation}\label{Wg formula 1}
\Wg (n,\sigma)=\prod_{c\in \sigma}\Wg (n,c)(1+O(n^{-2}))
\end{equation}
and 
\begin{equation}
\Wg (n,(1,\ldots ,d) ) = (-1)^{d-1}c_{d-1}\prod_{-d+1\leq j \leq d-1}(n-j)^{-1}
\end{equation}
where $c_i=\frac{(2i)!}{(i+1)! \, i!}$ is the $i$-th Catalan number.
\end{theorem}

As a shorthand for the
quantities in Theorem \ref{thm:mob}, we introduce the function $\Mob$ on the symmetric
group. $\Mob$ is invariant under conjugation and multiplicative over the cycles; further, it satisfies
for any permutation $\sigma \in \S_p$:
\begin{equation}\label{weingartenapprox}
\Wg(n,\sigma) = n^{-(p + |\sigma|)} (\Mob(\sigma) + O(n^{-2}))
\end{equation}
where $|\sigma |=p-\# \sigma $ is the \emph{length} of $\sigma$, i.e. the minimal number of transpositions that multiply to $\sigma$; we shall also use the notation $|\cdot |$ for the cardinality of sets. We refer to \cite{collins-sniady} for details about the function $\Mob$. We shall make use of the following explicit formulas. 

\begin{lemma}\label{lem:Mob}
The M\"{o}bius function is multiplicative with respect to the cycle structure of permutations
$$\Mob(\sigma) = \prod_{b \in \sigma} (-1)^{|b|-1} \Cat_{|b|-1},$$
where $\Cat_n$ is the $n$-th Catalan number. In particular, if $\sigma$ is a product of disjoint transpositions, then
\be\label{mob formula}
\Mob (\sigma) =(-1)^{|\sigma|}.
\ee
\end{lemma}

We finish this section by a well known lemma which we will use several times in the proofs of the main theorems of this paper. This result is contained in \cite{nica-speicher}.
\begin{lemma}\label{lem:S_p}
The function
$d(\sigma,\tau) = |\sigma^{-1} \tau|$ is an integer valued distance on $\S_p$. Besides, it has the following properties:
\begin{itemize}
\item the diameter of $\S_p$ is $p-1$;
\item $d(\cdot, \cdot)$ is left and right translation invariant;
\item for three permutations $\sigma_1,\sigma_2, \tau \in \S_p$, the quantity $d(\tau,\sigma_1)+d(\tau,\sigma_2)$
has the same parity as $d(\sigma_1,\sigma_2)$;
\item the set of geodesic points between the identity permutation $\id$ and some permutation $\sigma \in \S_p$ is in bijection with the set of non-crossing partitions smaller than $\pi$, where the partition $\pi$ encodes the cycle structure of $\sigma$. Moreover, the preceding bijection preserves the lattice structure. 
\end{itemize}
\end{lemma}

\subsection{Weingarten formulas for several independent unitary matrices} 

In this paper we want to treat random matrix models in which several i.i.d.\ Haar random unitary matrices appear. We shall use the independence property and the Weingarten formula together and unify these properties into an unique statement, as in Theorems \ref{thm:colored-Weingarten} and \ref{thm:colored-Weingarten-2}. Before we can do this, we introduce the following notation (we put $[n]=\{1,2,\ldots,n\}$):

\begin{definition}
Consider two integer functions $l:[p] \to [k]$ and $l':[p'] \to [k]$. Whenever $p=p'$, we denote by $\mathcal S_p^{l \to l'}$ the set of permutations of $p$ objects which map the level sets of $l$ to the level sets of $l'$:
$$\mathcal S_p^{l \to l'} = \{\sigma \in \mathcal S_p \, | \, \sigma(l^{-1}(s)) = \sigma(l'^{-1}(s))\}.$$
Note that this set is empty iff 
$$| \{t:l_t = s\} | \not = |\{t:l'_t = s\}|$$
for some $s \in [k]$. If $p \neq p'$, we put $\mathcal S_p^{l \to l'} = \emptyset$.
\end{definition}

Note that permutations $\sigma \in \mathcal S_p^{l \to l'}$ admit a decomposition along the level sets of $l,l'$
\begin{equation}\label{eq:decomp-perm-level-sets}
\sigma = \sigma_1 \times \cdots \times \sigma_k,
\end{equation}
where 
$$\sigma_s: l^{-1}(s) \to l'^{-1}(s).$$

The main result of this section is the following generalization of the Weingarten formula for moments in several independent Haar unitary matrices.

\begin{theorem}[generalized Weingarten formula]\label{thm:colored-Weingarten} 
Let $n$ and $k$ be positive integers and $i=(i_1,\ldots ,i_p)$, $i'=(i'_1,\ldots ,i'_{p'})$, $j=(j_1,\ldots ,j_p)$, $j'=(j'_1,\ldots ,j'_{p'})$ be tuples of positive integers from $[n]$ and $l=(l_1,\ldots ,l_p)$, $l'=(l'_1,\ldots ,l'_{p'})$ be tuples of positive integers from $[k]$. Then
\be\label{eq:colored-Weingarten} 
&&\int_{\mathcal{U}(n)^k} 
U_{i_1,j_1}^{(l_1)} \cdots U_{i_p,j_p}^{(l_p)} \overline{ U_{i'_1,j'_1}^{(l'_1)}} \cdots \overline{ U_{i'_{p'},j'_{p'}}^{(l'_{p'})}} 
dU^{(1)} \cdots dU^{(k)} \\ 
&=& \sum_{\alpha,\beta \in \mathcal S_p^{l \to l'}} 
\delta_{i_1 i'_{\alpha(1)}}\cdots\delta_{i_{p} i'_{\alpha(p)}} 
\delta_{j_1 j'_{\beta(1)}} \cdots \delta_{j_{p} j'_{\beta(p)}}
\Wg^{l \to l'}(n,\alpha, \beta),
\ee
where the modified Weingarten function $\Wg^{l \to l'}$ is defined via the product formula
$$ \Wg^{l \to l'}(n,\alpha, \beta) = \prod_{s=1}^k \Wg(n,\alpha_s^{-1}\beta_s).$$
\end{theorem}
\begin{proof}
We start by factoring the integral using the independence of the random unitary matrices $U^{(s)}$:
\be\label{Colored Weingarten proof}
&&\int_{\mathcal{U}(n)^k} 
U_{i_1,j_1}^{(l_1)} \cdots U_{i_p,j_p}^{(l_p)} \overline{ U_{i'_1,j'_1}^{(l'_1)}} \cdots \overline{ U_{i'_p,j'_p}^{(l'_p)}} 
dU^{(1)} \cdots dU^{(k)} \\
&=&
\int_{U(n)} [\text{factors with $l_t=1, l'_t =1$}] \, dU^{(1)} \times \cdots \times 
\int_{U(n)} [\text{factors with $l_t=k, l'_t =k$}] \, dU^{(k)} \notag
\ee
The above product vanishes whenever $|l^{-1}(s)| \neq |l'^{-1}(s)|$ for some $s \in [k]$. The value of each factor is computed using the usual Weingarten formula in Theorem \ref{thm:Wg} and the result follows.
\end{proof}

One can reformulate the above result in a more practical way, eliminating the restriction on the sum indices $\alpha, \beta$, as follows. 
 
\begin{theorem}[generalized Weingarten formula, second take]\label{thm:colored-Weingarten-2} 
Let $n$ and $k$ be positive integers and $i=(i_1,\ldots ,i_p)$, $i'=(i'_1,\ldots ,i'_{p'})$, $j=(j_1,\ldots ,j_p)$, $j'=(j'_1,\ldots ,j'_{p'})$ be tuples of positive integers from $[n]$ and $l=(l_1,\ldots ,l_p)$, $l'=(l'_1,\ldots ,l'_{p'})$ be tuples of positive integers from $[k]$. Then  
\begin{align}
\label{eq:colored-Weingarten-2}
&\qquad\int_{\mathcal{U}(n)^k} 
U_{i_1,j_1}^{(l_1)} \cdots U_{i_p,j_p}^{(l_p)} \overline{ U_{i'_1,j'_1}^{(l'_1)}} \cdots \overline{ U_{i'_p,j'_p}^{(l'_p)}} dU_1 \cdots dU_k \\ 
\notag
&= \sum_{\alpha, \beta \in S_p} 
\underbrace{
\delta_{i_1 i'_{\alpha(1)}}\cdots\delta_{i_{p} i'_{\alpha(p)}} 
\delta_{j_1 j'_{\beta(1)}} \cdots \delta_{j_{p} j'_{\beta(p)}}
}_{(\star)}
\underbrace{
\delta_{l_1 l'_{\alpha(1)}}\cdots\delta_{l_{p} l'_{\alpha(p)}} 
\delta_{l_1 l'_{\beta(1)}} \cdots \delta_{l_{p} l'_{\beta(p)}}
}_{(\star\star)} 
 \tilde \Wg (n, \alpha^{-1} \beta).
\end{align}
Here,
$$\tilde \Wg (n, \alpha^{-1} \beta) = \prod_{s=1}^k \Wg (n,\alpha_s^{-1} \beta_s) \quad \text{ when } \alpha, \beta \in \mathcal S_p^{l \to l'},$$
and it can have any other value when one of $\alpha$ or $\beta$ is not in $\mathcal S_p^{l \to l'}$.
\end{theorem}
\begin{proof}
This is a direct consequence of Theorem \ref{thm:colored-Weingarten}.
Indeed, because of the factors $(\star\star)$ in the above formula, 
only $\alpha, \beta \in \mathcal S_p^{l \to l'}$ survive, and the value of the Weingarten weight is the same as in Theorem \ref{thm:colored-Weingarten}. 
\end{proof}

For the concrete applications that follow, it is useful to have a simple equivalent for the quantity $\tilde \Wg (n, \alpha^{-1} \beta)$. 

\begin{proposition}\label{prop:equivalent-modified-Wg}
Assuming that $p$ is fixed, for every $l,l'$ and $\alpha, \beta \in \mathcal S_p^{l \to l'}$, we have
\begin{equation}\label{eq:Weingarten-function}
\Wg^{l \to l'}(n,\alpha, \beta) = \tilde \Wg (n, \alpha^{-1} \beta) =  \Wg (n, \alpha^{-1} \beta)(1+O(n^{-2}))
\end{equation}
\end{proposition}
\begin{proof}
We shall use the fact that one can approximate the usual Weingarten function by a polynomial times the M\"obius function, which are multiplicative. Start from the left hand side of the above equality, and use the fact that both $\alpha$ and $\beta$ decompose as products of smaller permutations acting on the level sets of $l$ and $l'$, see \eqref{eq:decomp-perm-level-sets}: 
\begin{align*}
\Wg^{l \to l'}(n,\alpha, \beta) &= \prod_{s=1}^k \Wg(n,\alpha_s, \beta_s)\\
&= \prod_{s=1}^k n^{-|l^{-1}(s)|-|\alpha_s^{-1}\beta_s|}\Mob(\alpha_s^{-1}\beta_s)(1+O(n^{-2}))\\
&=  n^{-p-|\alpha^{-1}\beta|}\Mob(\alpha^{-1}\beta)(1+ O(n^{-2}))\\
&= \Wg (n, \alpha^{-1} \beta)(1+O(n^{-2})).
\end{align*}

Note that the product over $k$ above has at most $p$ non trivial factors, because at most $p$ of the permutations $\alpha_s, \beta_s$ are non trivial. 

\end{proof}

\subsection{Graphical Weingarten calculus for several independent unitary matrices}

In this section we extend the graphical formalism introduced in \cite{cn1} to encompass integrals over several independent unitary Haar-distributed matrices.
We first recall briefly the single-matrix case. For more details on this method, we refer the reader to the paper \cite{cn1} and to other work which make use of this technique \cite{cn3,cn-entropy,cnz,cfn1}.

In the graphical calculus, matrices (or, more generally, tensors) are represented by boxes. 
To each box, one attaches symbols of different shapes, corresponding to vector spaces. 
The symbols can be empty (white) or filled (black), corresponding to spaces or their duals.

Wires connect these symbols, each wire corresponding to a tensor contraction $V \times V^* \to \mathbb C$. A diagram is a collection of such boxes and wires and corresponds to an element in a tensor product space.

The main advantage of such a representation is that it provides an efficient way of computing expectation values of such tensors when some (or all) of the boxes are random unitary matrices. In \cite{cn1}, the authors proposed an efficient way to apply the Weingarten formula \eqref{eq:Wg} to a diagram, which has a nice graphical interpretation.

The delta functions in each summand in the RHS of (\ref{eq:Wg}) describes how one should connect boxes. Each pair of permutations $(\alpha,\beta)$ in  (\ref{eq:Wg}) connects the labels of $U$ and $\overline U$ boxes and then, after erasing those boxes and keeping the wires, one obtains a new diagram. The permutation $\alpha$ is used to connect the white (or empty) labels of the boxes in the following way: the white decorations of the $i$-th $U$ box are connected to the corresponding white decorations of the $\alpha(i)$-th $\overline U$ box. Permutation $\beta$ is used in a similar manner to connect the black (or filled) decorations.

This process for a fixed permutation pair $(\alpha, \beta)$ is called a \emph{removal} and the whole process which sums all the new graphs over the all permutations is called the graph expansion. Importantly, the graphical calculus works linearly and separated components are related by tensor products, as is assumed implicitly above. The above procedure is summarized in the following important result. 

\begin{theorem}\label{thm:graphical-calculus}
If $\mathcal D$ is a diagram containing boxes $U, \overline U$ corresponding to a Haar-distributed random unitary matrix $U \in \mathcal U(n)$, the expectation value of $\mathcal D$ with respect to $U$ can be decomposed as a sum of weighted diagrams $\D_{\alpha, \beta}$ obtained by connecting the white labels of the $U$ boxed along the permutation $\alpha$ and the black labels along $\beta$. The weights of the diagrams are given by the Weingarten functions. 
\[\E_U(\D)=\sum_{\alpha, \beta} \D_{\alpha, \beta} \Wg (n, \alpha\beta^{-1}).\]
\end{theorem}

We now extend the graphical calculus by adding a new box corresponding to a column-block matrix made of independent unitary matrices and generalizing the removal procedure in order to allow to take expectations of diagrams containing such boxes. 

Consider i.i.d. Haar unitary matrices $U^{(1)},U^{(2)}, \ldots, U^{(k)} \in \mathcal U(n)$ and stack them up into a block-column $V \in M_{kn \times n}(\mathbb C)$
$$V = \sum_{i=1}^k e_i \otimes U^{(i)},$$
where $\{e_i\}$ is an orthonormal basis of $\mathbb C^k$ used to index the unitary blocks. Note that, up to a constant, $V$ is an isometry from $\mathbb C^n \to \mathbb C^k \otimes \mathbb C^n$, i.e. $V^*V=k\mathrm{I}_n$. Graphically, one can represent $V$ as in Figure \ref{fig:V-box}. With this correspondence, one can read the Weingarten formulas \eqref{eq:colored-Weingarten} and \eqref{eq:colored-Weingarten-2} in terms of the matrix $V$ by using the identification $U^{(l)}_{ij} = V_{(l,i),j}$. Moreover, the diagram of the (true) isometry $\tilde V$ used to defined random unitary channels can be easily obtained from $V$ and the diagonal matrix 
\begin{equation}\label{eq:W}
W = \mathrm{diag}(w_1, \ldots, w_k),
\end{equation}
see Figure \ref{fig:tilde-V-box}.
\begin{figure}[htbp]
\subfigure[]{\label{fig:V-box}\includegraphics{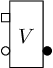}}\qquad\qquad
\subfigure[]{\label{fig:tilde-V-box}\includegraphics{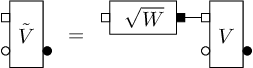}}\qquad\qquad
\subfigure[]{\label{fig:V-random}\includegraphics{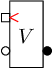}}
\caption{Boxes for block isometries. Round labels correspond to $\mathbb C^n$ and the square label corresponds to $\mathbb C^k$. On the left, the usual graphical representation of a stack of unitary operators. In the middle, the box $\tilde V$ for a weighted block-isometry. On the right, the box for a stack of i.i.d.\ random unitary matrices. The square label has a duplication symbol associated to it on the inside of the box, suggesting it should be connected using both permutations when performing the graph expansion.} 
\end{figure}

Consider now a diagram $\mathcal D$ containing a box $V$ corresponding to a stack of i.i.d. Haar unitary operators. We are interested in computing the expectation value of $\mathcal D$ with respect to $V$. In the spirit of Theorem \ref{thm:graphical-calculus}, we are going to express this expectation value as a sum over diagrams obtained by a new removal procedure, weighted by Weingarten coefficients. One of the main differences between the Weingarten formula for a single Haar unitary matrix \eqref{eq:Wg} and the Weingarten formula for several such Haar random matrices \eqref{eq:colored-Weingarten-2} is the fact that the indices $l$ and $l'$ in the latter equation are connected by \emph{both} permutations $\alpha$ and $\beta$. Hence, we need to redefine the removal procedure when computing expectation values with respect to $V$ boxes. Since the $l,l'$ indices are connected using both permutations, we shall add a \emph{duplication} symbol inside the box of $V$ of $V^*$, next to the label corresponding to $\mathbb C^k$, see Figure \ref{fig:V-random}. During the removal procedure, labels associated to the duplication symbols should be connected with both permutations $\alpha$ and $\beta$, as in Figure \ref{fig:expectation-block-isometry}.

\begin{figure}[htbp]
\includegraphics{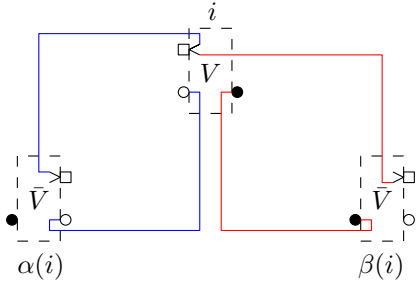}
\caption{Removal procedure for boxes corresponding to stacks of i.i.d. Haar unitary operators. Labels with duplication symbols must be connected using both permutations.} 
\label{fig:expectation-block-isometry}
\end{figure}

After the removal procedure, the duplication symbols remaining must be interpreted in terms of the duplication operator $T:\mathbb C^k \to \mathbb C^k \otimes \mathbb C^k$ and its adjoint $T^*$ (see Figure \ref{fig:duplication}):
\begin{align*}
T &= \sum_{i=1}^k e_i^* \cdot (e_i \otimes e_i), \\
T^* &= \sum_{i=1}^k (e_i^* \otimes e_i^*) \cdot e_i. 
\end{align*}

\begin{figure}[htbp]
\includegraphics{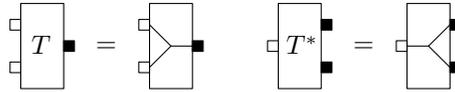}
\caption{Duplication operators $T$ and $T^*$.} 
\label{fig:duplication}
\end{figure}

\begin{theorem}\label{thm:V-graphical-calculus}
Let $\mathcal D$ be a diagram containing boxes $V, \bar V$ corresponding to a random matrices having a block-Haar distribution. The expectation value of $\mathcal D$ with respect to $V$ can be decomposed as a sum of weighted diagrams $\D_{\alpha, \beta}$ obtained by connecting the white labels of the $V$ boxed along the permutation $\alpha$, the black labels along $\beta$, and the labels with a duplication symbol along both $\alpha$ and $\beta$. The weights of the diagrams are given by modified Weingarten functions. 
\[\E_V(\D)=\sum_{\alpha, \beta} \D_{\alpha, \beta} \,\tilde \Wg (n, \alpha,\beta).\]
\end{theorem}

Several applications of this result shall be discussed in Sections \ref{sec:fixed-k} and \ref{sec:k-sim-n}.

\section{Haar random isometries versus block random isometries}\label{sec:measures-isometry}

In this section we perform a very basic analysis of the two probability measures we considered earlier on the set of isometries from $\mathbb C^n$ to $\mathbb C^k \otimes \mathbb C^n$. Let us first recall their definitions. The easiest way to introduce these measures is via an image measure construction, starting from the Haar measure on unitary groups. Recall that both these objects are probability measures on the set of isometries 
$$\mathrm{Isom}(\mathbb C^n,\mathbb C^{kn}) = \{V \in M_{kn \times n}(\mathbb C) \, | \, V^*V = \mathrm{I}_n\}.$$

The Haar measure $\mu_\text{Haar}$ is the easiest to define: it is the image measure of the Haar measure on the unitary group $\mathcal U(kn)$ via the truncation operation which erases the last $(k-1)n$ columns of a $kn \times kn$ matrix. In other words, if $P \in M_{kn \times n}(\mathbb C)$ is the truncation operator $P_{ij} = \delta_{ij}$ and $U \in \mathcal U(kn)$ is a Haar-distributed unitary matrix, then $UP \sim \mu_\text{Haar}$.

An isometry $\tilde V \in \mathrm{Isom}(\mathbb C^n,\mathbb C^{kn})$ has distribution $\mu_{ w}$ if it is obtained by stacking $k$ independent Haar unitary matrices one on top of the other, with weights $\sqrt{w_i}$:
\begin{equation}\label{eq:V-def}
\tilde V = \sum_{i=1}^k \sqrt{w_i} e_i \otimes U^{(i)}.
\end{equation}
The probability distribution $\mu_{ w}$ can also be seen as the image measure of the product of $k$ Haar measures on $\mathcal U(n)$ via the weighted stacking procedure described above.

We gather next some basic properties of these measures, whose proofs are left to the reader.
\begin{proposition}
The measures $\mu_\text{Haar}$ and $\mu_\text{block}$ have the following invariance properties:
\begin{enumerate}
\item If $V \sim \mu_\text{Haar}$ and $U \in \mathcal U(kn)$, $U' \in \mathcal U(n)$ are fixed unitary matrices, then $UVU' \sim \mu_\text{Haar}$.
\item If $V \sim \mu_{ w}$ and $U' \in \mathcal U(n)$ is a fixed unitary matrix, then $VU' \sim \mu_{ w}$.
\item If $V \sim \mu_{ w}$ and $U \in \mathcal U(kn)$ is such that
$$U = \sum_{j=1}^k e_{\sigma(j)}e_j^* \otimes U^{(j)},$$
where $\sigma \in \mathcal S_k$ is a permutation that leaves the vector $w$ invariant and $U^{(j)} \in \mathcal U(n)$ are unitary matrices, then $UV \sim \mu_{ w}$.
\end{enumerate}
\end{proposition}

One can easily discriminate statistically between the two measures by computing moments or covariances for different matrix entries. In what follows, for the sake of simplicity, we shall consider the equiprobability vector $w_*=(1/k, \ldots, 1/k)$.
Let us start by computing the moments of a single matrix element. It is well known (see \cite{hiai-petz}) that, for a Haar unitary matrix $U$ of size $n$, one has
$$\mathbb E |U_{11}|^{2p} = \binom{n+p-1}{n-1}^{-1}.$$
It follows that, for the two ensembles we consider, we have
\begin{align*}
\mathbb E_\text{Haar}  |V_{11}|^{2p} &= \binom{kn+p-1}{kn-1}^{-1},\\
\mathbb E_{w_*}  |V_{11}|^{2p} &= k^{-p} \binom{n+p-1}{n-1}^{-1}.
\end{align*}

Although the above expressions agree at $p=1$, they are different at $p=2$, showing a statistical difference between the two ensembles. 

More striking examples come from covariance computations: matrix elements $V_{11}$ and $V_{1, n+1}$ are independent under $\mu_{w_*}$, while this is obviously not true for $\mu_\text{Haar}$, see \cite[Proposition 4.2.3]{hiai-petz}.

\section{Product of conjugate channels with bounded output dimension}
\label{sec:fixed-k}

We start by representing random unitary channels (and the corresponding complementary channel) in the graphical formalism we introduced. In terms of the random block-Haar map $V$, the channel has the following form:
$$\Phi(X) = \sum_{i=1}^k w_i U_i X U_i^* = \mathrm{Tr}_k ((\sqrt W \otimes \mathrm{I}_n) VXV^*(\sqrt{W} \otimes \mathrm{I}_n)),$$
where $X \in M_n(\C)$ is the input matrix and $\tilde V = (\sqrt{W} \otimes \mathrm{I}_n) V$ is the isometry in the Steinespring picture. As before, we define the weighting matrix $W = \mathrm{diag}(w_1, w_2,\ldots w_k)$ and $V$ is obtained by stacking the unitary matrices $U_i$ one on top of the other, as in equation \eqref{eq:V-def}. Graphically, the diagram corresponding to the channel $\Phi$ is presented in Figure \ref{fig:Phi}, whereas the complementary channel is depicted in Figure \ref{fig:Phi-c}.

Next, we want to describe the limiting output eigenvalues of a fixed input state going through a random unitary channel. We are interested in the following $k^2 \times k^2$ random matrix:
\be
Z_n =[\Phi^C \otimes \overline{\Phi}^C] (\psi_n\psi_n^* ) 
\ee
Here, $\psi_n \in \mathbb C^n \otimes \mathbb C^n$ is a fixed input vector for each $n$; notice that we are considering only rank-one inputs, since these states are known to yield minimal entropy outputs.

\begin{figure}[htbp]
\subfigure[]{\label{fig:Phi}\includegraphics{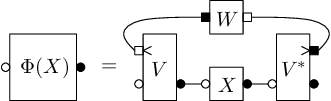}}\qquad\qquad
\subfigure[]{\label{fig:Phi-c}\includegraphics{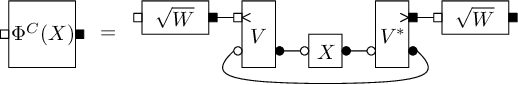}}
\caption{A random unitary channel and its complementary.} 
\end{figure}

To represent the input $\psi_n$ in the graphical calculus, we add $A_n$ and $A_n^*$ boxes on the wires of the Bell input, as in Figure \ref{AAstar}. 
\begin{figure}[htbp]
 \begin{center}
  \includegraphics[]{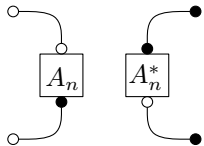} 
 \end{center}
 \caption{Generalized Bell states are used as inputs.} 
 \label{AAstar}
\end{figure}

Algebraically, we consider a sequence of inputs
\be\label{input-formula}
\psi_n = \sum_{i,j =1}^{n} a_{ij} e_i \otimes e_j,
\qquad a_{ij} \in \C.
\ee
In the matrix version, this reads
\be
A_n = \sum_{i,j =1}^{n} a_{ij} e_i \otimes e_j^*.
\ee
We require the normalization relation $\|\psi_n\|=1$, which is equivalent to 
\be 
\trace [A_n A_n^*] =1.
\ee

The input vectors $\psi_n$ generalize Bell states
$$\phi_n = \frac{1}{\sqrt{n}} \sum_{i=1}^n e_i \otimes e_i,$$
which correspond to the trivial choice $A_n = \mathrm{I}_n / \sqrt{n}$.

\subsection{Calculation of the limiting eigenvalues} 
In order to define \emph{well-behaved} inputs, 
as in \cite{cfn1},
we introduce two assumptions on the asymptotic behavior of the sequence of input states $A_n$.
\\
{\bf Assumption 1:}
\be\label{eq:ass1}
\frac{\trace \left[ A_n \right] }{\sqrt{n}} = m + O\left(\frac{1}{n^2}\right)
\ee
for some $m\in\C$.
Note that a similar relation holds for $A_n^*$, with $\bar m$ replacing $m$. Note that one has 
$$m = \lim_{n \to \infty} \langle \psi_n, \phi_n \rangle,$$
so that one can say that $m$ (or rather $|m|^2$) measures the overlap between the input state $\psi_n$ and the Bell state $\phi_n$.
\\
{\bf Assumption 2:} 
\begin{equation}\label{eq:ass2}
\|A_n\|_\infty = O\left(\frac{1}{\sqrt{n}}\right)
\end{equation}

Recall that the empirical eigenvalue distribution of a self-adjoint matrix $Z \in M_{k^2}(\mathbb C)$ is the probability measure
$$k^{-2} \sum_{i=1}^{k^2} \delta_{\lambda_i},$$
where $\lambda_1, \ldots, \lambda_{k^2}$ are the eigenvalues of $Z$.

Before we state our result, we introduce one more notation, essential to what follows. 
\begin{definition}\label{def:S}
Let $S:\mathbb R^2 \times \mathbb R^k \to \mathbb R^k$ be the function defined by
$$S(x,y; w) = \mathrm{spec}^\downarrow(H_\Sigma(x,y;w)),$$
where $\mathrm{spec}^\downarrow$ denotes the ordered spectrum of a self-adjoint matrix and 
\begin{equation}\label{eq:H-Sigma}
\forall i,j \in [k], \qquad H_\Sigma(i,j)=
\begin{cases}
(x+y)w_i^2 \qquad &\text{ if } i=j\\
y w_i w_j \qquad &\text{ if } i \neq j.
\end{cases}
\end{equation}
\end{definition}

\begin{theorem}
\label{thm:bell-phenomenon}
Under the assumptions \eqref{eq:ass1} and \eqref{eq:ass2} above, the empirical eigenvalue distribution of the matrix $Z_n$ converges \emph{almost surely}, as $n \rightarrow \infty$, to the probability measure:
\be\label{limitdistUbar}
\frac{1}{k^2} \left [\sum_{\substack{i,j=1\\ i \neq j }}^k \delta_{w_iw_j} + \sum_{i=1}^k \delta_{s_i} \right],
\ee
where $s = S(1-|m|^2, |m|^2 ; w)$. 
\end{theorem}

Before we prove this theorem, let us state some of its corollaries and analyze the limit entropy of the matrix $Z_n$ (which is the entropy of the probability vector appearing in the conclusion of the theorem) as a function of the parameters $|m|^2$ and $w_i$.

We analyze first the ``extremal'' cases for the weight vector $w$. 
\begin{corollary}
In the case where the weighting vector is uniform,  $W=\mathrm{I}/k$, the $S$ function can be evaluated to give
\be
s_1 = \frac{|m|^2}{k} + \frac{1-|m|^2}{k^2} ; \qquad s_i = \frac{1-|m|^2}{k^2} \quad (2 \leq i \leq k).
\ee
This implies that the output state has asymptotically the following eigenvalues:
\begin{itemize}
\item $\frac{|m|^2}{k} + \frac{1-|m|^2}{k^2}$, with multiplicity one;
\item $\frac{1-|m|^2}{k^2}$ with multiplicity $k-1$;
\item $\frac{1}{k^2}$, with multiplicity $k^2-k$.
\end{itemize}
The entropy of the probability vector
$$\left( \frac{|m|^2}{k} + \frac{1-|m|^2}{k^2}, \underbrace{\frac{1-|m|^2}{k^2}, \ldots, \frac{1-|m|^2}{k^2}}_{k-1 \text{ times}} , \underbrace{\frac{1}{k^2}, \ldots, \frac{1}{k^2}}_{k^2-k \text{ times}} \right)$$
is a decreasing function of $|m|^2$, the asymptotic overlap between the input vector $\psi_n$ and the Bell state $\phi_n$.
\end{corollary} 
\begin{corollary}
In the case where the weighting vector is trivial, $w=(1,0,\ldots, 0)$ the channel $\Phi$ is a unitary conjugation and the output matrix $Z_n$ is a pure state of null entropy. 
\end{corollary} 

We now turn to the ``extreme" values of the parameter $m$, $|m|=1$ (the input state is, up to a phase, a Bell state) and $m=0$ (the input state is orthogonal to the Bell state).

\begin{corollary}
In the case where the input state is equal, up to a phase, to the Bell state, i.e.\ $|m|^2=1$, the matrix $H_\Sigma$ is, up to a constant, a rank one projector and thus 
$$s_1 = \sum_{i=1}^k w_i^2, \qquad s_2 = \cdots = s_k = 0.$$ 
\end{corollary} 
\begin{corollary}
In the case where the input state is orthogonal to the Bell state, i.e.\ $m=0$, the matrix $H_\Sigma$ is diagonal and thus 
$$s_i=w_i^2, \qquad \forall i \in [k].$$
The limiting eigenvalue vector of the output state $Z_n$ is $w \otimes w$ and its entropy is thus 
$$\lim_{n \to \infty} H(Z_n) = H(w \otimes w) = 2H(w).$$
\end{corollary} 

\begin{proof}[Proof of the Theorem \ref{thm:bell-phenomenon}]
The proof uses the moment method and consists of two steps. First, we compute the asymptotic moments of the output density matrix $Z_n$ and then, by a Borel-Cantelli argument, we deduce the almost sure convergence of the spectral distribution and of the eigenvalues. 

{\bf Step 1:} We calculate the limit moments of $Z_n$, using the graphical calculus, see Figure \ref{fig:Zn}. Here, $\includegraphics{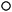}$ and $\includegraphics{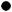}$ correspond to the $n$-dimensional space, and $\includegraphics{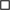}$ to the $k$-dimensional output space.

\begin{figure}[htbp]
 \begin{center}
  \includegraphics{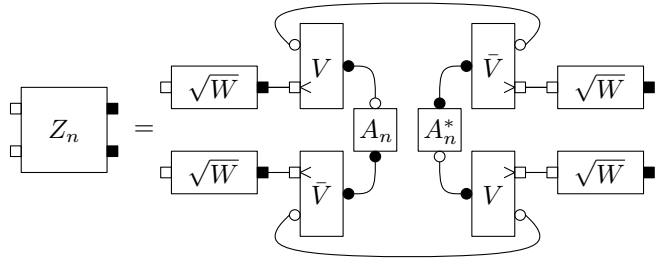} 
 \end{center}
 \caption{The diagram for the output state $Z_n$.} 
 \label{fig:Zn}
\end{figure}

In order to compute the $p$-th moment of the matrix $Z_n$, we use the graphical calculus on a diagram obtained by connecting $p$ copies of the diagram in Figure \ref{fig:Zn} in a tracial manner. For fixed $p \in \N$ the Weingarten sum in Theorem \ref{thm:V-graphical-calculus} are indexed by pairs of permutations $(\alpha, \beta) \in \S_{2p}^2$. We label the $V$ and $\ol V$ boxes in the following manner: 
$1^T, 2^T, \ldots, p^T$ for the $V$ boxes of the first channel (T as ``top'') and 
$1^B, 2^B, \ldots, p^B$ for the $V$ boxes of the second channel (B as ``bottom''). 
We shall also order the labels as $\{1^T, 2^T, \ldots, p^T, 1^B, 2^B, \ldots, p^B\} \isom \{1, \ldots, 2p\}$. 
A removal $r=(\alpha, \beta) \in \S_{2p}^2$ of the $V$ and $\ol V$ boxes connects the decorations in the following way:
\begin{enumerate} 
\item the round white decoration \includegraphics{circle_w.eps} of the $i$-th $V$-block is paired with the round white decoration of the $\alpha(i)$-th $\ol V$ block by a wire;
\item the round black decoration \includegraphics{circle_b.eps} of the $i$-th $V$-block is paired with the round black decoration of the $\beta(i)$-th $\ol V$ block by a wire;
\item the square white decoration \includegraphics{square_w.eps} of the $i$-th $V$-block is paired with both the square white decorations of the $\alpha(i)$-th and $\beta(i)$-th $\ol V$ blocks by wires. This double pairing is suggested by the duplication symbol associated to the square label. 
\end{enumerate}
We also introduce two fixed permutations $\gamma, \delta \in \S_{2p}$ which represent wires appearing in the diagram before the graph expansion.
The permutation $\gamma$ represents the initial wiring of the $\includegraphics{square_w.eps}$ decorations (corresponding to the trace operation) and $\delta$ accounts for the wires between the $\includegraphics{circle_b.eps}$ decorations connecting boxes $A$ or $A^*$.
More precisely, for all $i$,
\begin{equation} 
\label{eq:def-gamma-delta}
\gamma(i^T) = (i-1)^T, \quad \gamma(i^B) = (i+1)^B,\quad
and \quad
\delta(i^T) = i^B, \quad \delta(i^B) = i^T.
\end{equation}
After the removal procedure, for each pair of permutations $(\alpha,\beta)$, we obtain a diagram $\mathcal D_{\alpha,\beta}$ consisting of:
\begin{enumerate}
	\item $\includegraphics{circle_w.eps}$-\emph{loops}; $n^{\#(\alpha)}$
	\item $\includegraphics{square_w.eps}$-\emph{nets};	$f_W(\alpha,\beta)$
	\item $\includegraphics{circle_b.eps}$-\emph{necklaces}; $f_A(\beta)$
\end{enumerate}
First, one can easily see that the number of $\includegraphics{circle_w.eps}$-loops is exactly $n^{\# \alpha}$.
Next, for $f_W(\alpha,\beta)$,  
since square labels are connected by wires with the box $T$'s and $T^*$'s, 
the graph they yield is not a collection of loops, but can be more general,
where the boxes $W$'s are "caught in nets" which are made of $T$'s and $T^*$'s.
The general formula for $f_W$ can be found in Lemma \ref{W-in-nets}.  
Finally, the contribution of $\includegraphics{circle_b.eps}$-necklaces depends on the moments of the matrices $A_n$ and is encoded in a function $f_A(\beta)$ (see \cite{cfn1} for a more detailed treatment of a similar situation):
\be
f_A(\beta) &=& \prod_{c \in {\rm Cycle}(\beta^{-1}\delta)} 
\trace \left[ A^{s_{c,1}} \cdots A^{s_{c,\|c\|}}  \right]
\ee
Here, $|c|$ is the number of elements in $c$ and 
$s_{c,1} \ldots ,s_{c,|c|}$ are defined such that
\be 
s_{c,i}=
\begin{cases}
1 & \text{if the $i$th element in the cycle $c$ belongs to $T$}\\
* & \text{if the $i$th element in the cycle $c$ belongs to $B$}
\end{cases}
\ee
Note that the above function $f(\beta)$ is well-defined in spite of the ambiguity of $s_{c,i}$, because of the circular property of the trace.

Therefore, the Weingarten formula in Theorem \ref{thm:V-graphical-calculus} reads
\be \label{moments}
\E \trace [Z_n^p] 
= \sum_{\alpha,\beta \in S_{2p}} n^{\# \alpha } f_W(\alpha,\beta) f_A(\beta) \tilde\Wg(n,\alpha,\beta).
\ee
Using the moment growth assumptions \eqref{eq:ass1}, \eqref{eq:ass2} for the matrices $A_n$, we get that, for all cycle $c$ of $\beta$, we have
\be
&&|\trace \left[ A^{s_{c,1}} \cdots A^{s_{c,|c|}}  \right]| 
\leq 
\| A^{s_{c,1}} \|_\infty \cdots \| A^{s_{c,|c| -1}}\|_\infty \cdot \| A^{s_{c,|c|}} \|_1 \\ \notag
&&\lesssim
\left(\frac{1}{\sqrt{n}} \right)^{\|c\|-1} \cdot\sqrt{n}
= n^{1-|c|/2} 
\qquad \text{as $n \rightarrow \infty$},
\ee
where the notation $f(x) \lesssim g(x)$ means that there exists some constant $C>0$ such that $f(x) \leq Cg(x)$ for $x$ large enough. The above inequality is the only place where Assumption 2 (see (\ref{eq:ass2})) is used. 
Hence it yields the following asymptotic bound for the factor $f(\beta)$;
\be\label{roughbound_f}
|f(\beta)|  \lesssim   n^{\#(\beta^{-1}\delta) -p }.
\ee

Using the equivalent for the (modified) Weingarten function in Proposition \ref{prop:equivalent-modified-Wg}, 
we get (note that the factors depending on the fixed parameter $k$ are hidden in the $\lesssim$ notation)
\bee
 \E \trace [Z_n^p] 
\lesssim \sum_{\alpha,\beta \in S_{2p}} n^{\#\alpha} 
n^{\#(\beta^{-1}\delta) -p} 
n^{-2p - |\alpha^{-1}\beta|} 
\qquad \text{as $n \rightarrow \infty$} .
\label{bound_f} 
\eee
The power of $n$ in the RHS of (\ref{bound_f}) 
is bounded by using Lemma \ref{lem:S_p} as follows.
\bee
2p- |\alpha| +p-|\beta^{-1}\delta| -2p - |\alpha^{-1}\beta| 
= p- (|\alpha| + |\alpha^{-1}\beta| + |\beta^{-1}\delta|) \leq 0
\eee 
Here, equality holds if and only if 
$\id \rightarrow \alpha \rightarrow \beta \rightarrow \delta$ is a geodesic:
\be\label{eq:geodesic-a-b}
\alpha = \prod_{i\in A} \tau_i, \qquad \beta = \prod_{i\in B} \tau_i
\ee
where $\tau_i = (i^T,i^B)$ and $A\subseteq B\subseteq \{1, \ldots,p\}$; we refer to
\cite{cn1} for a proof of this fact.
Importantly, for such geodesic permutations $\beta$, the following asymptotic behavior follows form the first assumption on the growth of the trace of $A_n$:
\bee\label{behaviour_f}
f_A (\beta) =  \left(n|m|^2\right)^{|\beta|}  + O \left(n^{-2}\right).  
\eee
Note that $|\beta|=|B|$.
This implies that the power of $n$  in (\ref{moments}) in fact becomes $0$
for all the $\alpha,\beta$ which satisfy the geodesic condition  
$\id \rightarrow \alpha \rightarrow \beta \rightarrow \delta$:
\be
\#\alpha +  |B| -2p - |\alpha^{-1}\beta| =
2p - |A| + |B|  -2p - |B \setminus A| =0
\ee 
Here,  $|B| = |B \setminus A|+ |A|$.

Hence, by using (\ref{weingartenapprox}), (\ref{mob formula}) and (\ref{behaviour_f}),
we have the following approximation on $\E \trace [Z_n^p] $ 
(note that the estimate on the error order is not necessary here but will be so in Step 2):
\be\label{momentsapprox}
\Ex  \left[Z_n^p\right] 
=  \left(1 + O \left(n^{-2}\right)\right)
\sum_{\id \rightarrow \alpha \rightarrow \beta \rightarrow \delta}
f_W(\alpha,\beta)|m|^{2|B|} 
 (-1)^{|B \setminus A|}  
\ee 
For the above error term, note that Lemma \ref{lem:S_p} implies that 
permutations $(\alpha,\beta)\in S_{2p} \times S_{2p}$ off the geodesic 
make the power of $n$ less by two or more; only even powers are allowed. 

Using Lemma \ref{lem:counting-loops}, we can further process the moment expression
\be 
&&\lim_{n\rightarrow \infty}  \E \trace [Z_n^p] 
= \left(\trace [W^p]\right)^2 \\
&+&\underbrace{\sum_{A=\emptyset, B \not = \emptyset}  \trace \left[W^{2p}\right] |m|^{2|B|}  (-1)^{|B|} }_{(\spadesuit)} 
+\underbrace{\sum_{\emptyset \not =A \subseteq B} 
\prod_{i=1}^{|A|} \trace \left[ W^{2(a_{i+1}-a_i)} \right] |m|^{2|B|}  (-1)^{|B \setminus A|} }_{(\heartsuit)}. \notag
\ee
The multinomial identities:
\begin{align*}
\sum_{\emptyset \subseteq A \subseteq \{1, \ldots, p\}} \!\!\!\! x^{|A|} &= (1+x)^p \\
\sum_{\emptyset \subseteq A \subseteq B \subseteq \{1, \ldots, p\}} \!\!\!\!\!\!\!\! x^{|A|}y^{|B \setminus A|} &= (1+x+y)^p.
\end{align*}
give further calculations
\be
(\spadesuit) = \trace [W^{2p}] \left[\left(1-|m|^2\right)^p -1\right].
\ee
Set $C=B\setminus A$ such that $\emptyset \subseteq C \subseteq [p]\setminus A$ and use Lemma \ref{technical-identity} with $x=1-|m|^2$ and $y=|m|^2$ to get
\be
(\heartsuit) &=& \sum_{ A \not = \emptyset} \prod_{i=1}^{|A|} \trace \left[ W^{2(a_{i+1}-a_i)} \right] 
\sum_{\emptyset \subseteq C \subseteq [p]\setminus A} |m|^{2(|C|+|A|)}  (-1)^{|C|} \\
&=& \sum_{A \not = \emptyset} \prod_{i=1}^{|A|} \trace \left[ W^{2(a_{i+1}-a_i)} \right]
|m|^{2|A|} \left(1-|m|^2 \right)^{p-|A|} \\
&=&\sum_{i=1}^k s_i^p - (1-|m|^2)^p\trace\left[W^{2p} \right],
\ee
where, by Lemma \ref{technical-identity} and Definition \ref{def:S}, the vector $s \in \mathbb R^k$ is given by $s = S(1-|m|^2, |m|^2 ; w)$. Therefore, as a whole,
\be\label{the limit}
\lim_{n\rightarrow \infty}  \E \trace [Z_n^p] 
&=& \left(\trace [W^p]\right)^2 + (\spadesuit) + (\heartsuit)\\
&=& \left(\trace [W^p]\right)^2 - \trace [W^{2p}] 
+\sum_{i=1}^k s_i^p  \\
&=& \sum_{\substack{i,j=1 \\ i \neq j}}^k (w_iw_j)^p + \sum_{i=1}^k s_i^p,
\ee
 expression in which one can recognize the limiting eigenvalues announced in the theorem.

{\bf Step 2:} 
We now move on to prove the almost sure convergence.
Since this part of proof is very similar to that of Theorem 6.3 in \cite{cn1}
or Theorem 3.1 in \cite{cfn1},
we only sketch here the main ingredients. Using Borel-Cantelli Lemma, it is enough to prove that
the covariance series converges:
\bee\label{series}
\sum_{n=1}^\infty \Ex \left[
\left(\trace \left[ Z_n^p\right] - \Ex \trace \left[ Z_n^p \right] \right)^2\right]
= \sum_{n=1}^\infty \Ex \left[
\left(\trace \left[ Z_n^p\right] \right)^2 \right]- \left(\Ex \trace \left[ Z_n^p \right] \right)^2
 < \infty
\eee
which will imply that for all $p\geq 1$   
\be
\trace \left[ Z_n^p\right]   \rightarrow (\ref{the limit}) 
\quad \text{a.e.} \quad \text{as $n \rightarrow \infty$} .
\ee 

Also, note that by Carleman's condition, equation
(\ref{the limit}) uniquely determines the measure as in (\ref{limitdistUbar}).

First, (\ref{momentsapprox}) implies that 
\begin{align}\label{moment^2}
\left(\Ex \trace \left[ Z_n^p \right] \right)^2
=  
\left(k^{-2p} \sum_{\id \rightarrow \alpha \rightarrow \beta \rightarrow \delta}
f_W(\alpha,\beta)|m|^{2|\beta|} 
 (-1)^{|\alpha^{-1}\beta|}\right)^2 
+ O(n^{-2}) 
\end{align}

On the other hand, we use Theorem \ref{thm:V-graphical-calculus} to calculate $\Ex [(\trace \left[ Z_n^p] \right)^2 ]$.
In the diagram we have two identical copies of $\trace Z_n^p$, which amounts to a total of $4p$ pairs of $V$ and $\overline V$ boxes.
As a result, 
removals $(\bar\alpha,\bar\beta)$ are defined for $\bar\alpha,\bar\beta \in S_{4p}$. 
However, importantly those two copies are initially separated. 
Namely, initial wires $\bar \gamma, \bar \delta \in S_{2p}\oplus S_{2p} = S_{4p}$ 
are written as direct sums:
\be
\bar \gamma = \gamma_1 \oplus \gamma_2 
\quad \text{and} \quad 
\bar\delta = \delta_1 \oplus \delta_2,
\ee
where the indices 1 and 2 refer to the first or the second group of $2p$ boxes appearing in the diagram.
Then, as before, we calculate the power of $n$, which is 
\be
2p- (|\bar\alpha| + |\bar\alpha^{-1}\bar\beta| + |\bar\beta^{-1}\bar\delta|) \leq 0.
\ee
Equality holds if and only if 
$\id \rightarrow \bar\alpha \rightarrow \bar\beta \rightarrow \bar\delta$
is a geodesic. 
Moreover, this geodesic condition implies that
$\bar\alpha$ and $\bar\beta$ can be written as
\be
\bar\alpha = \alpha_1 \oplus \alpha_2 
\quad \text{and} \quad
\bar\beta =   \beta_1 \oplus \beta_2
\ee 
Here, 
pairs $(\alpha_1, \beta_1)$ and $(\alpha_2, \beta_2)$ are 
defined as in (\ref{eq:geodesic-a-b}).

Therefore, in the diagram all removals which matter as $n\rightarrow \infty$
keep those two copies separated. 
Also, these removals have the following properties:
\begin{align*}
& f_W(\bar\alpha, \bar\beta )
=f_W(\alpha_1,\beta_1)
\times f_W(\alpha_2,\beta_2), \\
&|\bar\beta| = |\beta_1|+|\beta_2|, \qquad
|\bar\alpha^{-1} \bar\beta |= |\alpha_1^{-1} \beta_1 | + |\alpha_2^{-1} \beta_2 |  
\end{align*}
The first statement says that $f_W$ can be calculated for the each copy independently. Then, for the same reasons as before, we get an approximation with the error of order $1/n^2$:
\begin{align}\label{2ndmoment}
&\Ex  \left(\trace\left[ Z_n^p\right] \right)^2 =  O(n^{-2}) \\
&+  k^{-4p} 
\sum_{\substack{\id \rightarrow \alpha_1 \rightarrow \beta_1 \rightarrow \delta \\
\id \rightarrow \alpha_2 \rightarrow \beta_2 \rightarrow \delta}}
\left[
f_W(\alpha_1,\beta_1) f_W(\alpha_2,\beta_2)
|m|^{2|\beta_1|+2|\beta_2|}
(-1)^{|\alpha_1^{-1} \beta_1 | + |\alpha_2^{-1} \beta_2 |  } \right]. \notag
\end{align}

Finally, we see from \eqref{moment^2} and \eqref{2ndmoment} that
\bee
\Ex \left[\left(\trace \left[ Z_n^p\right] \right)^2 \right]- \left(\Ex \trace \left[ Z_n^p \right] \right)^2
 = O(n^{-2}) 
\eee
which proves (\ref{series}), and finalizes the proof of the theorem. 
\end{proof}

The following rather technical lemmas are needed in the proof of the result above when dealing with nets containing $W$ boxes. 

\begin{lemma}\label{W-in-nets}
For fixed permutations $\alpha, \beta \in \mathcal S_{2p}$
\be
f_W(\alpha,\beta) = \prod_{b \in (\gamma^{-1}\alpha \vee \gamma^{-1}\beta)} \trace W^{|b|}
\ee
Here, we understand the notation $\gamma^{-1}\alpha \vee \gamma^{-1}\beta$ in terms of partitions: both permutations $\gamma^{-1}\alpha$ and $\gamma^{-1}\beta$ naturally induce partitions on $\{1,\ldots,2p\}$, 
and $\vee$ is the join operation on the poset of (possibly crossing) partitions.
Also, $b$ and $|b|$ stand for a block of a partition and its cardinality, respectively. 
\end{lemma} 
\begin{proof}
First, we observe that because of the cyclic structure of $Z_n$
we can put two $\sqrt{W}$-boxes into one $W$-box and associate it to the neighboring $U$-box. 
Next, we claim that since $W$ is a diagonal matrix, 
we can slide $W$ along wires within the net.
Indeed, we can show it algebraically:
\be \quad
\sum_i (e_i^* W) \otimes e_i \otimes e_i = \sum_i w_i e_i^* \otimes e_i \otimes e_i 
=\sum_i e_i^*  \otimes (W e_i) \otimes e_i = \sum_i e_i^*  \otimes e_i \otimes (W e_i).
\ee
Finally, by using the above fact, we collect the $W$ matrices together. 
However, the rest which is composed of $T$'s and $T^*$'s can be contracted to a point. 
Therefore, each connected component $b$ in $\gamma^{-1}\alpha \vee \gamma^{-1}\beta$ 
gives the factor $\trace W^{|b|}$. 
\end{proof}

The above general formula must be studied in details in order to complete the proof of  Theorem \ref{thm:bell-phenomenon}. We need precise values of $f_W(\alpha,\beta))$ when $\alpha,\beta$ lie on the geodesic $\id \rightarrow \alpha \rightarrow \beta \rightarrow \delta$.

\begin{lemma}\label{lem:counting-loops}
Suppose permutations $\alpha, \beta$ lie on the geodesic $\id \rightarrow \alpha \rightarrow \beta \rightarrow \delta$. As in \eqref{eq:geodesic-a-b}, they admit decompositions as products of disjoint transpositions indexed by subsets $A \subseteq B$.
Then:
\begin{itemize}
\item  When $A = \emptyset$ we have
\be
f_W(\id,\beta) = 
\begin{cases} \left[\trace \left( W^{p}\right)\right]^2 & \text{ if } B=\emptyset \\ 
\trace \left( W^{2p}\right) & \text{ if } B\not=\emptyset.
\end{cases} 
\ee
\item When $A = \{a_1 < \cdots < a_{|A|} \} \not = \emptyset$ we have
\be
f_W(\alpha,\beta) = \prod_{i=1}^{|A|} \trace \left[ W^{2(a_{i+1} -a_i)} \right].
\ee
Here, we understand $a_{|A|+1} -a_{|A|}$ to be equal to $p+a_1 - a_{|A|}$.
\end{itemize}
\end{lemma}
\begin{proof}
First, we consider the case $A = \emptyset$. 
If $B = \emptyset$, the net associated to $W$ is composed of two cycles containing each $p$ $W$ boxes (the top and the bottom cycles). However, if $B$ is nonempty, the top and bottom cycles become connected and one obtains a large cycle of length $2p$. In conclusion, these cases yield respectively $f_W(\id,\id)=[\trace (W^p)]^2$ and $f_W(\id,\beta)=\trace(W^{2p})$,
by using Lemma \ref{W-in-nets}.  

Next, let us assume that $A = \{a_1, \ldots, a_{|A|} \} \not = \emptyset$. 
The structure of $\alpha$ and $\beta$ as in \eqref{eq:geodesic-a-b} with $A \subseteq B$ implies that the connected components of the net are determined by indices $i \in [p]$ such that $i \in A$ and $i \in B$, see Figure \ref{fig:net-W} for a proof. Since the first condition implies the second and each such connected components carries $2(a_{i+1} -a_i)$ boxes $W$, one obtains the announced formula.

\begin{figure}[htbp]
\subfigure[]{\includegraphics{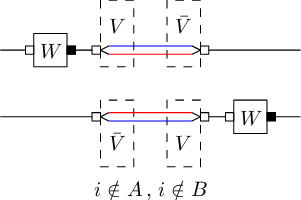}}\\
\subfigure[]{\includegraphics{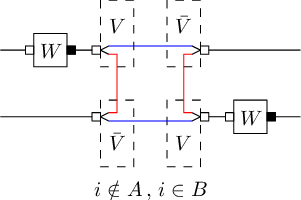}}\\
\subfigure[]{\includegraphics{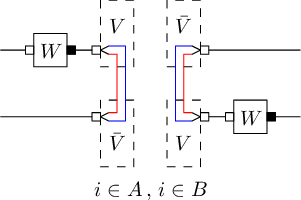}}\\
\caption{The connected components of the net containing $W$ boxes corresponding to a pair of geodesic permutations $\alpha,\beta$ is determined by indices $i \in A$. In the picture, only the last case induces a ``cut'' in the cycles, creating additional connected components. } 
\label{fig:net-W}
\end{figure}

\end{proof}

The following result is needed to simplify the formulas in the lemma above. It can be, however, interesting on its own, from a combinatorial perspective. 

\begin{lemma}\label{technical-identity}
For a real diagonal matrix $W = \mathrm{diag}(w_1,\ldots,w_k)$ and real numbers $x,y \in \mathbb R$, we have
\be
\sum_{\emptyset \not= A \subseteq [p]} x^{p-|A|}y^{|A|} \prod_{i=1}^{|A|} \trace \left[W^{2(a_{i+1}-a_i)}\right] 
= \sum_{i=1}^k s_i^p - x^p\trace\left[W^{2p} \right],
\ee
where $A=\{a_1 < \cdots < a_{|A|}\}$, $(a_{|A|+1}-a_{|A|})$ should be understood as $(p+a_1 - a_{|A|})$ and $s=S(x,y;w)$ is as in Definition \ref{def:S}.
\end{lemma}
\begin{proof}
To the diagonal operator $W = \sum_{i=1}^k w_i e_i e_i^*$, we associate the vector
$$\mathbb R^k \otimes \mathbb R^k \ni \tilde W = \sum_{i=1}^k w_i e_i \otimes e_i$$
and $P_{\tilde W} \in M_{k^2}(\mathbb C)$, the orthogonal projection on $\tilde W$. The idea of the proof is to consider the self-adjoint operator
$$H = x W \otimes W + y P_{\tilde W}$$
and to expand $\mathrm{Tr} (H^p)$. Since $xW \otimes W$ and $yP_{\tilde W}$ do not commute in general, one has to consider general words in these two matrices. Such words can be indexed by the positions $A=\{a_1 < \cdots < a_{|A|}\}$ where $P_{\tilde W}$ appears in the word; let us call $\mathcal W_A$ the word corresponding to a subset $A \subseteq [p]$. Oviously, one has $\mathrm{Tr}(\mathcal W_\emptyset) = x^p[\mathrm{Tr}(W^p)]^2$, and thus 
$$\mathrm{Tr}(H^p) = x^p[\mathrm{Tr}(W^p)]^2 + \sum_{\emptyset \neq A \subseteq [p]} \mathrm{Tr}(\mathcal W_A).$$
Using the graphcal notation, it follows from Figure \ref{fig:trace-H-p}, that the general term $\mathrm{Tr}(\mathcal W_A)$ factorizes along the intervals defined by the set $A$ and one has, for all $A \neq \emptyset$,
\begin{equation}\label{eq:trace-W-A}
\mathrm{Tr}(\mathcal W_A) = x^{p-|A|}y^{|A|} \prod_{i=1}^{|A|} \trace \left[W^{2(a_{i+1}-a_i)}\right].
\end{equation}

\begin{figure}[htbp]
\subfigure[]{\includegraphics{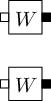}}\qquad\qquad
\subfigure[]{\includegraphics{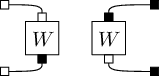}}\\
\subfigure[]{\label{fig:trace-H-p-interval}\includegraphics{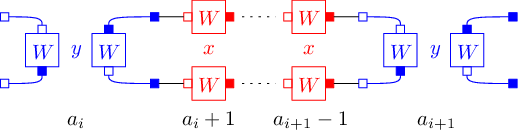}}\\
\caption{Diagrams for $W \otimes W$, $P_{\tilde W}$, and for an interval inside a word $\mathcal W_A$.} 
\label{fig:trace-H-p}
\end{figure}

More precisely, each conected component in the diagram for $\mathcal W_A$ corresponds to an interval $a_{i+1}-a_i$ and, for such a trace (see Figure \ref{fig:trace-H-p-interval}), one has a contribution of $x^{a_{i+1}-a_i-1}y\mathrm{Tr}\left[W^{2(a_{i+1}-a_i)}\right]$. Multiplying all these contribution gives \eqref{eq:trace-W-A}.

It follows that 
\begin{equation}\label{eq:sum-A-intermediate}
\sum_{\emptyset \not= A \subseteq [p]} x^{p-|A|}y^{|A|} \prod_{i=1}^{|A|} \trace \left[W^{2(a_{i+1}-a_i)}\right] 
= \trace\left( H^p \right) - x^p\left[\trace (W^p)\right]^2.
\end{equation}
One can further simplify this equation by using the explicit form of the operator $H$. Indeed, either by using the graphical notation or by simple algebra, the action of $H$ on some basis vectors can be easily computed as follows. For $i \neq j$, it is obvious that $P_{\tilde W}e_i \otimes e_j = 0$ and thus $H e_i \otimes e_j = xw_iw_j e_i \otimes e_j$, proving that $e_i \otimes e_j$ are eigenvectors of $H$ for the eigenvalues $x w_i w_j$. 

Let $\Sigma = \oplus_{i=1}^k \mathbb R e_i \otimes e_i$ be the subspace orthogonal to the space spaned by $e_i \otimes e_j$ with $i \neq j$. The restriction of $H$ to $\Sigma$ is exactly the operator $H_\Sigma$ defined in \eqref{eq:H-Sigma}, Definition \ref{def:S}. Thus, the eigenvalue vector of $H_\Sigma$ is $s = S(x,y;w)$. We now have computed all the $k^2$ eigenvalues of $H$ and we have
\begin{equation}\label{eq:trace-H-p}
\mathrm{Tr}(H^p) = \sum_{\substack{i,j=1\\ i \neq j }} x^p (w_iw_j)^p + \sum_{i=1}^p s_i^p.
\end{equation}
The conlcusion of the lemma follows now easily from \eqref{eq:sum-A-intermediate}, \eqref{eq:trace-H-p}, and the following equality
$$\left[\mathrm{Tr}(W^p)\right]^2= \sum_{\substack{i,j=1\\ i \neq j }} (w_iw_j)^p +\sum_{i=1}^k w_i^{2p}.$$

\end{proof}

\section{Product of conjugate channels with unbounded output dimension}
\label{sec:k-sim-n}

In this section, we consider the case where the output dimension grows with the input dimension of the channel, in a linear manner:
$$k/n \rightarrow c,$$ 
where $c$ is a positive constant that we consider as a parameter of the model. Since both $n$ and $k$ grow to infinity, there is no incentive to consider complementary channels, so we focus on the output of the original channels
 \be
Z_n = [\Phi \otimes \overline{\Phi}] ({\psi_n}{\psi_n^*} ),
\ee
where $\psi_n$ is the generalised Bell state introduced in the previous section. We shall make the same assumptions \eqref{eq:ass1},\eqref{eq:ass2} on the growth of the matrices $A_n$ appearing in the definition of $\psi_n$. Moreover, since the number of unitary matrices in $\Phi$ grows with $n$, we introduce the following assumptions on the growth of the  weight matrices $W_n$. 
\\
{\bf Assumption 3:}
\begin{equation}\label{eq:ass3}
\forall p \geq 1, \qquad \lim_{n \to \infty} \frac{1}{k} \trace \left[(kW_n)^p \right] = t_p = \int x^p d\nu(x).
\end{equation} 
where $t_p$ are the moments of some given compactly supported measure $\nu$. The probability measure $\nu$, or, equivalently, the moment sequence $(t_p)_{p \geq 1}$, are parameters of the model and they are fixed. The trace-preserving condition for the channel $\Phi$, $\trace(W_n) = 1$, implies that $t_1 = 1$.

We first compute the moments of the $n^2 \times n^2$ output matrix $Z_n$.

\begin{theorem}\label{thm:moments-ZnC}
Under the assumptions \eqref{eq:ass1}, \eqref{eq:ass2}, and \eqref{eq:ass3}, the output matrix $Z_n$ has the following asymptotic moments:
\be
\Ex \trace [(cnZ_n)^2] &=& t_2^2 + c^2 + t_2^2|m|^4 + O(n^{-1}) ; \\
\Ex \trace [(cnZ_n)^p] &=& t_2^p|m|^{2p} + O(n^{-1}), \qquad \forall p \geq 3.
\ee
\end{theorem}
\begin{proof}
We start by applying the graphical expansion procedure described in Theorem \ref{thm:V-graphical-calculus} to the diagram for $\Ex \trace (Z_n^p)$, obtained by connecting $p$ copies of the diagram for $Z_n$, displayed in Figure \ref{fig:ZnC} in a tracial way. We obtain a formula which is very close to \eqref{moments}, the only differences coming from the fact that we are not using complementary channels. The notation is the same as the one in the proof of Theorem \ref{thm:bell-phenomenon}.

\begin{figure}[htbp]
 \begin{center}
  \includegraphics{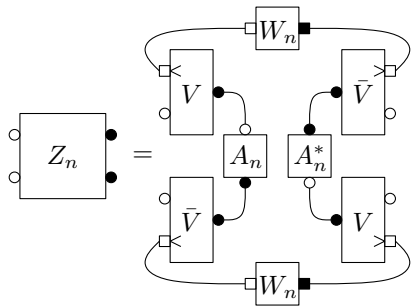} 
 \end{center}
 \caption{The diagram for the output state $Z_n$.} 
 \label{fig:ZnC}
\end{figure}

\be \label{eq:moments-Wn}
\Ex \trace [(cnZ_n)^p]
= (cn)^{p} \sum_{\alpha,\beta \in S_{2p}} n^{\# (\gamma^{-1} \alpha) } f_W(\alpha,\beta) f_A(\beta)  \tilde\Wg(n,\alpha,\beta).
\ee

Let us first upper bound the factor $f_W(\alpha,\beta)$. The boxes $W_n$ appearing in a diagram $\mathcal D_{\alpha,\beta}$ are connected in a net whose connected components are given by the blocks of the partition $\alpha \vee \beta$. 
In the spirit of Lemma \ref{lem:counting-loops}, each such connected component $b$ contributes a factor of $\mathrm{Tr}(W_n^{|b|})$, which, by \eqref{eq:ass3}, is equivalent to $k^{1-|b|}t_{|b|}$. Hence, 
\begin{equation}\label{eq:equiv-f-W}
f_W(\alpha,\beta) \sim k^{\#(\alpha \vee \beta) - 2p} \prod_{b \in \alpha \vee \beta} t_{|b|}.
\end{equation}
Replacing $k \sim cn$ and using the bound \eqref{roughbound_f} for $f_A(\beta)$, we can express everything in terms of $c$ and $n$. The power of $n$ appearing in the general term of the moment formula can then be bounded by
\begin{align*}
& p+ \# (\gamma^{-1} \alpha) + \#(\beta^{-1}\delta) -p + \#(\alpha \vee \beta) -2p -2p- |\alpha^{-1}\beta| \\
 &\qquad \leq 2p - (|\gamma^{-1}\alpha|+|\alpha^{-1}\beta|+|\beta^{-1}\delta|+|\beta| )   \\
& \qquad \leq 2p - (|\gamma^{-1}\beta|+|\beta^{-1}\delta|+|\beta|  ) \leq 0.
\label{S_1}
\end{align*}
Here, the first inequality holds because 
\be
\#(\alpha \vee \beta) \leq \# \beta
\label{ab-relation} 
\ee
the second one is true by the triangle inequality (Lemma \ref{lem:S_p})
\be
|\gamma^{-1}\alpha|+|\alpha^{-1}\beta| \geq |\gamma^{-1}\beta|.
\label{bg-line}
\ee
Finally, the last inequality follows from the proof of Theorem 6.8 in \cite{cn3}, which we recall as a lemma.
\begin{lemma}\label{b-value}
For any permutation $\beta \in \mathcal S_{2p}$, one has
\be
|\gamma^{-1}\beta|+|\beta^{-1}\delta|+|\beta| \geq 2p
\ee
equality holding if and only if 
\be
\beta = \begin{cases} 
\id, \delta, \gamma & p =2; \\  \delta & p \geq 3.
\end{cases} 
\ee
\end{lemma}

Next, we analyze when the power of $n$ becomes $0$. 
When $\beta = \delta$, the condition (\ref{ab-relation}) enforces $\alpha$ to be on the geodesic: 
$\id \rightarrow \alpha \rightarrow \beta = \delta$. 
However, with the equality condition for (\ref{bg-line}), we conclude that $\alpha = \delta$.
It is easy to see that $\beta = \id$ implies $\alpha = \id$ via (\ref{ab-relation}) and
that $\beta = \gamma$ results in $\alpha = \gamma$ because of the equality condition for (\ref{bg-line}). 
One can easily check that for these values of $\beta$, the bound \eqref{roughbound_f} is saturated
from the calculation below.

Since we identified the dominating terms in the moment equation \eqref{eq:moments-Wn}, it is now easy to compute the limits; 
we just plug the following equivalents into \eqref{eq:moments-Wn}.
\begin{align*}
\alpha =\beta&= \delta &f_W(\delta, \delta) &\sim (cn)^{-p} t_2^p
&f_A(\delta) &\sim n^p |m|^{2p} \\
\alpha =\beta&= \id &f_W(\mathrm{id}, \mathrm{id}) &= 1
&f_A(\mathrm{id}) &= 1\\
\alpha =\beta&= \gamma & f_W(\gamma, \gamma) &\sim (cn)^{2-2p}t_p^2
&f_A(\gamma) &= 1.
\end{align*}

When $p \geq 3$, we get
\be
\Ex \trace [(cnZ_n)^p]=  t_2^p|m|^{2p} + O(n^{-1}),
\ee
while when $p=2$, we obtain
\be
\Ex \trace [(cnZ_n)^2]   = t_2^2 + c^2 + t_2^2|m|^4 + O(n^{-1}).
\ee
\end{proof}

The fact that two different behaviours appear in the limits above depending on the value of $p$, can be explained by the presence of eigenvalues on different scales. We first apply the Hayden-Winter trick \cite{hayden-winter, hastings} to our  weighted random unitary channel setting. Note that the following proposition applies to any random unitary channel and any input state, at fixed dimension $n$.
\begin{proposition}\label{prop:hayden-winter-trick}
Consider the output of a rank one input state through a product of conjugated random unitary channels
$$Z_n = [\Phi \otimes \bar \Phi] (\psi_n \psi_n^*),$$
where the vector $\psi_n \in \mathbb C^n \otimes \mathbb C^n$ is defined as in \eqref{input-formula} and the channel $\Phi$ has weights as in \eqref{eq:def-Phi}.Then, one has the following lower bound for the largest eigenvalue of $Z_n$:
$$\lambda_1(Z_n) \geq \frac{|\mathrm{Tr} A_n|^2}{n} \sum_{i=1}^k w_i^2.$$
\end{proposition} 
\begin{proof}
After expanding the sums, one has
$$Z_n = \sum_{i,j=1}^k w_i w_j (U_i \otimes \bar U_j) \psi_n \psi_n^* (U_i \otimes \bar U_j)^*.$$
If $\varphi_n$ is the vector corresponding to the Bell state
\begin{equation}\label{eq:Bell-state}
\phi_n = \frac{1}{\sqrt n} \sum_{i=1}^n e_i \otimes e_i,
\end{equation}
one has
$$\langle \varphi_n, Z_n \varphi_n \rangle \geq \sum_{i=1}^k w_i^2 \langle \varphi_n, (U_i \otimes \bar U_i) \psi_n \psi_n^* (U_i \otimes \bar U_i) \varphi_n \rangle.$$
Using the fact that, for all unitary transformations $U$, one has $(U \otimes \bar U)\varphi_n = \varphi_n$, we get
$$\lambda_1(Z_n) \geq\langle \varphi_n, Z_n \varphi_n \rangle \geq  |\langle \varphi_n, \psi_n \rangle|^2 \sum_{i=1}^k w_i^2  = \frac{|\mathrm{Tr} A_n|^2}{n} \sum_{i=1}^k w_i^2.$$
\end{proof} 

This Theorem \ref{prop:hayden-winter-trick} gives the following lemma as a corollary:
\begin{lemma}\label{prop:hayden-winter-trick2}
Take input states $A_n$ satisfying assumption \eqref{eq:ass1} and weights $W_n$ satisfying the scaling \eqref{eq:ass3}. 
Then, for all realizations of the random matrix $V$, the largest eigenvalue of the matrix $Z_n$ can lower bounded, asymptotically, as follows:
$$\liminf_{n\rightarrow \infty}\lambda_1(cnZ_n) \geq t_2 |m|^2.$$
\end{lemma} 
\begin{proof}First, Theorem \ref{prop:hayden-winter-trick} implies that
\be
cn \lambda_1(Z_n) \geq \frac{|\mathrm{Tr} A_n|^2}{n} cn \trace W^2
\ee
Then, we use \eqref{eq:ass1} and \eqref{eq:ass3} with $p=2$. 
\end{proof} 

Note that this behaviour is consisted with the moments computed in Theorem \ref{thm:moments-ZnC}, for $p \geq 3$. In order to investigate the smaller eigenvalues, we analyse the matrix $Q_nZ_nQ_n$, where $Q_n=I_n- \tilde E_n$. 
Here, $\tilde E = \phi_n \phi_n^*$ is the projection on Bell state \eqref{eq:Bell-state}.
Precise statements about the spectrum will be made later, using Cauchy's interlacing theorem. Before we state our theorem, let us recall the definition of a compound free Poisson distribution.

Compound free Poisson distributions were introduced in \cite{sp2} by Speicher and the theory was further developed in \cite{hiai-petz} and \cite[Prop. 12.15]{nica-speicher}. Traditionally, they are defined via a limit theorem that mimics the classical Poisson limit theorem. 

\begin{definition}\label{def:compound-Poisson}
Let $\lambda$ be a positive real number and $\mu$ a compactly supported probability measure. The limit in distribution, as $N \to \infty$ of the probability measure
$$ \left[ \left( 1-\frac{\lambda}{N} \right) \delta_0 + \frac{\lambda}{N}\mu \right]^{\boxplus N}$$
is called a \emph{compound free Poisson distribution} of rate $\lambda$ and jump distribution $\mu$ ,
which is denoted by $\pi_{\lambda,\mu}$; it has free cumulants given by
$$\kappa_p(\pi_{\lambda,\mu}) = \lambda m_p(\mu),$$
where $m_p(\mu)$ denotes the $p$-th moment of the probability distribution $\mu$.
\end{definition}
Note that the usual free Poisson (or Marchenko-Pastur) distributions $\pi_c$ are special cases of the above definition, obtained by letting $\lambda = c$ and $\mu = \delta_1$. We also introduce the notation for the distribution of the square of a random variable: if $X$ has distribution $\mu$, then $\mu^{\times 2}$ is the distribution of the random variable $X^2$. It follows that the moments of $\mu^{\times 2}$ are 
$$m_p(\mu^{\times 2})  = m_p(\mu)^2.$$

Before going to the the main result of this section, 
we make a remark on free cumulants. 
The free cumulants and moments of a random variable satisfy the following relation, called the moment-cumulant formula: 
\be\label{cumulant-formula}
m_p = \sum_{\sigma \in NC(p)} \prod_{b \in \sigma} \kappa_{|b|}
\ee
Here, $NC(p)$ is the non-crossing partition and $b \in \sigma$ is a block of the partition $\sigma$.  
For more details, please see, for example, \cite[Proposition 1.4]{nica-speicher}.

\begin{theorem}\label{thm:compound-Poisson-limit}
The empirical eigenvalue distribution of the matrix $(cn)^2Q_nZ_nQ_n$ converges in moments  to a compound free Poisson distribution with rate $c^2$ and jump distribution $\nu^{\times 2}$. 
\end{theorem} 
\begin{proof}
This proof is similar with the one of Theorem 6.10 in \cite{cn3} and uses the method of moments.
\be
\frac{1}{n^2} \Ex \trace ((cn)^2Q_nZ_nQ_n)^p 
&=& c^{2p} n^{2p-2} \Ex \trace \prod_p (I_n-\tilde E_n )Z_n \\
&=& c^{2p} n^{2p-2} \sum_{g \in \mathcal F_p} (-1)^{|g^{-1}(E_n)|} n^{-|g^{-1}(E_n)|} 
\underbrace{\Ex \trace \prod_{i=1}^p g(i) Z_n}_{(\diamondsuit)}
\ee
Here, $E_n = n \tilde E_n$ and $g \in \mathcal F_p = \{h: \{1,2,\ldots,p\} \rightarrow \{I_n, E_n\}\}$. 
To calculate $(\diamondsuit)$, 
we need to set the natural correspondence of $g$ in $S_{2p}$, which is denoted by $\hat g$:
When $i \in g^{-1}(I)$,
\be
\hat g ((i+1)^T) = i^T \quad \text{and} \quad \hat g (i^B) = (i+1)^B
\ee
and when $i\in g^{-1}(E)$,
\be
\hat g ((i+1)^T) = (i+1)^B \quad \text{and} \quad \hat g (i^B) = i^T
\ee
With this notation we have
\be
(\diamondsuit) = \sum_{\alpha,\beta \in S_{2p}} 
f_W(\alpha,\beta) n^{\# (\hat g^{-1} \alpha)} f_A(\beta) \tilde \Wg(\alpha^{-1}\beta)
\ee
Putting everything together and interchanging the two sums, we obtain
\begin{align}
\frac{1}{n^2} \Ex & \trace ((cn)^2Q_nZ_nQ_n)^p  = \\
&= c^{2p} n^{2p-2} \sum_{g \in \mathcal F_p} (-1)^{|g^{-1}(E_n)|} n^{-|g^{-1}(E_n)|} 
\sum_{\alpha,\beta \in S_{2p}} 
f_W(\alpha,\beta) n^{\# (\hat g^{-1} \alpha)} f_A(\beta) \tilde \Wg(\alpha^{-1}\beta) \\
&= c^{2p} n^{2p-2} 
\sum_{\alpha,\beta \in S_{2p}} 
f_W(\alpha,\beta)f_A(\beta) \tilde \Wg(\alpha^{-1}\beta)
\underbrace{\sum_{g \in \mathcal F_p}  (-1)^{|g^{-1}(E_n)|}n^{-|g^{-1}(E_n)| +\# ( \hat g^{-1} \alpha)}}_{(\clubsuit)}
\end{align}

Importantly, it was shown in \cite{cn3} that $(\clubsuit)$ vanishes unless $\alpha$ belongs to the following set:
\be
\tilde S_{2p} = \{\pi \in S_{2p} : \text{$\pi \delta$ has no fixed point} \}
\ee
Moreover, it follows also from \cite{cn3} that for such $\alpha \in \tilde S_{2p}$ one has 
\be\label{bound-ad}
|\alpha \delta| \geq p.
\ee

With the estimates \eqref{eq:equiv-f-W} and \eqref{roughbound_f} for $f_W$ and $f_A$ respectively, we have
\begin{align*}
 \frac{1}{n^2} &\Ex \trace ((cn)^2Q_nZ_nQ_n)^p \lesssim  c^{2p} n^{-p-2} \times  \\
&\sum_{\substack{\alpha \in \tilde S_{2p} \\ \beta \in S_{2p} \\g \in \mathcal F_p } }
(cn)^{\# (\alpha \vee \beta)-2p} \; t_{\alpha \vee \beta} \;
n^{\# (\beta^{-1}\delta) -|\alpha^{-1}\beta| -|g^{-1}(E_n)| +\# (\hat g^{-1} \alpha)} \;
(-1)^{|g^{-1}(E_n)|} \;
\Mob(\alpha^{-1}\beta),
\end{align*}
where $t_{\alpha \vee \beta}$ is defined multiplicatively over the cycles of the partition $\alpha \vee \beta$:
$$t_{\alpha \vee \beta} = \prod_{b \in \alpha \vee \beta} t_{|b|}.$$
The power of $n$ in the expression above is bounded by the triangle inequality and (\ref{bound-ad}):
\be
&& p-2  + \#(\alpha \vee \beta) -|\beta^{-1}\delta| -|\alpha^{-1}\beta| -|g^{-1}(E_n)| - |\hat g^{-1} \alpha| \\
&\leq& 3p -2 - (|\alpha| + |\beta^{-1}\delta| + |\alpha^{-1}\beta| + |g^{-1}(E_n)| + |\hat g^{-1} \alpha| )  \\
&\leq& 3p -2 - (|\alpha| + |\alpha^{-1}\delta| + |g^{-1}(E_n)| + |\hat g^{-1} \alpha| )  \\
&\leq& 2p -2 - (|\alpha|  + |g^{-1}(E_n)| + |\hat g^{-1} \alpha| )  \\
&\leq& 2p -2 - ( g^{-1}(E_n) + |\hat g^{-1}| ) 
\ee
In the above each $\leq$ is $=$ respectively if and only if
\be \notag
{\rm(i)}\, \# (\alpha \vee \beta) = \# \alpha, \quad
{\rm(ii)}\, \alpha \rightarrow \beta \rightarrow \delta, \quad
{\rm(iii)}\,|\alpha^{-1}\delta| =p \quad
{\rm(iv)}\, \id \rightarrow \alpha \rightarrow \hat g 
\ee

We also claim that
\be
2p -2 - ( g^{-1}(E_n) + |\hat g^{-1}| ) \leq 0
\ee
and $\leq$ becomes $=$ if and only if $g=I$, called the condition (v). 
This is from the following fact:
\be
\# \hat g = \begin{cases} 2 & g \equiv I \\ |g^{-1}(E_n)| & g \not\equiv I \end{cases} 
\ee
Hence we conclude that the power of $n$ becomes $0$ if and only if
the conditions (i) to (v) are satisfied. 

The condition (v) implies that $\hat g = \gamma$,
which in turn via (iv) forces $\alpha$ to have the following structure:
\be\label{structure-a}
\alpha = \alpha^T \oplus \alpha^B 
\ee
Here, $\id \rightarrow \alpha^T \rightarrow \gamma^T$ and
$\id \rightarrow \alpha^B \rightarrow \gamma^B$
where $\gamma^T = (p^T, (p-1)^T, \ldots, 1 )$ and 
$\gamma^B = (1^B,2^B, \ldots, p^B)$.
This result also enforces a similar structure to $\beta$ via (i):
\be\label{structure-b}
\beta = \beta^T \oplus \beta^B 
\ee
Moreover, the special structure of $\delta$ shows that 
$\# ((\alpha^T \oplus \alpha^B)\delta) = \#(\alpha^T \alpha^B)$.
Here, for the RHS, we identify $i^T = i^B$. 
Then, (iii) imposes $\# (\alpha^T\alpha^B)=p$, i.e., $\alpha^T\alpha^B = \id$.

The condition (iii) with $\alpha \in \tilde S_{2p}$ implies that
$\alpha^{-1} \delta$ is paring. 
However, note that $\alpha$ moves each point within T-group and B-group,
and on the other hand, $\delta$ moves each point between these groups. 
So, $\alpha^{-1} \delta$ has the following structure:
\be
\alpha^{-1} \delta = \prod_{i \in \Lambda} \tau_i 
\ee
Here, $\Lambda$ is an index set and $\tau_i = (a_i^T, b_i^B)$. 
Now, we reinterpret (ii) as
\be
\id \rightarrow \alpha^{-1} \beta \rightarrow \alpha^{-1} \delta
\ee
which implies, by using (\ref{structure-a}) and (\ref{structure-b}), 
that $\alpha = \beta$. 

Hence we now have three conditions:
\be
{\rm(v)}\, g \equiv I,\quad {\rm(vi)}\, \alpha = \beta, \quad 
{\rm(vii)}\, \alpha = \alpha^T \oplus \alpha^B 
\text{ with $\id \rightarrow \alpha^T \rightarrow \gamma^T$ and $\alpha^T\alpha^B = \id$}
\ee
and check that these conditions really gives the leading power.
The condition (v) implies that $|g^{-1}(E_n)| =0$ and 
the conditions (vi) and (vii) lead to
\be
f_A(\beta) = f_B(\alpha) = 1
\ee
Note that the conditions (v), (vi) and (vii) are in fact 
neccessary and surficient for the power of $n$ to be $0$
because $f(\beta)$ achieves the bound \eqref{roughbound_f}
under these conditions. 

Therefore, 
\begin{align*}
\lim_{n \to \infty} \frac{1}{n^2} &\Ex \trace ((cn)^2Q_nZ_nQ_n)^p\\
&=\sum_{\alpha^T \text{ in (vii)}} 
c^{2\#(\alpha^T)}t_{\alpha^T \oplus \alpha^B}\Mob(\id) 
= \sum_{\sigma \in NC(p)} c^{2\#\sigma}m^2_\sigma(\nu) \\
&=\sum_{\sigma \in NC(p)} \prod_{b \in \sigma} c^2 m_p (\nu^{\times 2}) =m_p(\pi_{c^2,\nu^{\times 2}}).
\end{align*}
Here, we used the fact that $\alpha\in S_p$ such that $\id \rightarrow \alpha \rightarrow (1,2,\ldots,p)$ 
corresponds to a non-crossing partion of $[p]$. 
The last equality above follows from \eqref{cumulant-formula} and
the moment-cumulant formula for the compound Poisson distribution of rate $c^2$ and jump distribution $\nu^{\times 2}$, finalising the proof. 
\end{proof} 
 
\begin{theorem}\label{thm:k-sim-n}
The ordered eigenvalues $\lambda_1 \geq \lambda_2 \geq \ldots \geq \lambda_{n^2}$ of the output matrix $Z_n$ have the following asymptotic behaviour, as $n \to \infty$:
\begin{enumerate}
\item In probability, $cn \lambda_1 \rightarrow t_2|m|^2$.
\item The empirical eigenvalue distribution $\frac{1}{n^2-1} \sum_{i=2}^{n^2} \delta_{(cn)^2\lambda_i}$ converges weakly to
the compound free Poisson distribution $\pi_{c^2,\nu^{\times 2}}$ of rate $c^2$ and jump distribution $\nu^{\times 2}$. 
\end{enumerate}
Here, $\nu$ is a probability measure defined in \eqref{eq:ass3}.
\end{theorem} 
\begin{corollary}
In the case where the weighting vector is uniform,  $W_n=\mathrm{I}/k$, the limiting measure is the Dirac mass $\nu = \delta_1$. The largest eigenvalue of $Z_n$ behaves as $|m|^2/(cn)$ and the asymptotic shape of the lower spectrum is given by a free Poisson distribution of parameter $c^2$, $\pi_{c^2}$. Then, almost surely, one obtains the following asymptotic behavior for the entropy of the output matrix $Z_n$ (see \cite[Proposition 6.12]{cn3} for a proof):
$$H(Z_n) = 
\begin{cases}
2\log n - \frac{1}{2c^2}+o(1) \qquad &\text{ if } c\geq 1,\\
2\log(cn) - \frac{c^2}{2}+o(1) \qquad &\text{ if } 0<c<1.
\end{cases}$$
\end{corollary} 

\begin{proof}[Proof of Theorem \ref{thm:k-sim-n}]
First, let $\tilde \lambda_1 \geq \ldots \geq \tilde \lambda_{n^2-1}$ be
the ordered eigenvalues of $Q_nZ_nQ_n$. By Cauchy's interlacing theorem (for example, see Corollary III.1.5 of \cite{bhatia}) one has
\be
\lambda_1 \geq \tilde \lambda_1 \geq \lambda_2 \geq \ldots \geq \lambda_{n^2-1} \geq\tilde\lambda_{n^2-1} \geq \lambda_{n^2},
\ee
which, together with the conclusion of Theorem \ref{thm:compound-Poisson-limit}, proves the second statement. 

Next, we prove the first part. 
Since $ c |\trace A_n|^2 \trace W^2 \rightarrow t_2 |m|^2 $ as $n\rightarrow \infty$,
we define 
\be
\epsilon_n =\left| c |\trace A_n|^2 \trace W^2 - t_2 |m|^2  \right|  
\ee
Then, 
Proposition \ref{prop:hayden-winter-trick} implies that 
\be
1\leq \frac{cn \lambda_1}{t_2|m|^2 - \epsilon_n} \leq \frac{(cn \lambda_1)^3 }{(t_2|m|^2 - \epsilon_n)^3} 
\leq \frac{\trace[ (cnZ_n)^3]}{(t_2|m|^2 - \epsilon_n)^3}.
\ee 
After taking expectations and using Theorem \ref{thm:moments-ZnC}, we obtain
\be
1 \leq  \frac{\Ex [cn \lambda_1]}{t_2|m|^2 - \epsilon_n} \leq \frac{t_2^3|m|^6 +O(n^{-1}) }{(t_2|m|^2 - \epsilon_n)^3},
\ee
Then, there exists a sequence of positive numbers $\{\epsilon_n^\prime\}$ such that 
$\epsilon_n^\prime \rightarrow 0$ as $n \rightarrow \infty$, and
\be
\Ex [cn \lambda_1] \leq t^2 |m|^2 + \epsilon_n^\prime
\ee
Markov's inequality implies that for any $\delta >0$, 
\be
\P(cn\lambda_1 - t_2 |m|^2 \geq \delta )  
&\leq&  \P(cn\lambda_1 - t_2 |m|^2 + \epsilon_n \geq \delta  )  \\
&\leq& \frac{\Ex [cn\lambda_1 - t_2 |m|^2 + \epsilon_n]}{\delta} = \frac{\epsilon_n^\prime + \epsilon_n}{\delta}
\ee
The other bound is obvious from Lemma \ref{prop:hayden-winter-trick2}. 
\end{proof}

\section{Conclusions and final remarks}\label{sec:conclusions}

In this final section, we would like to compare the results obtained in the current paper with similar results for non unit-preserving random quantum channels studied in \cite{cn1, cn3, cfn1}. 

Before going into details, let us first note that the weights appearing in the definition of random unitary channels \eqref{eq:def-Phi}  represent more parameters that can be chosen to one's convenience. This explains why the limiting objects in the current paper are more general than the ones in \cite{cn1, cn3, cfn1}. 

Let us first analyze the case where $k$ is fixed. We shall compare the results in Theorem \ref{thm:bell-phenomenon} to the ones in \cite[Theorem 3.1]{cfn1}. To do this, we must first make sure that the channels we compare have the same input and output spaces. We must thus enforce the condition $t=1/k$ in \cite[Theorem 3.1]{cfn1}. We find that the $k^2$ limiting eigenvalues are \emph{different} in the two cases:
\begin{align*}
\lambda^{RC}(Z_n) &\to \left( \frac{|m|^2}{k} + \frac{1}{k^2}-\frac{|m|^2}{k^3}, \underbrace{\frac{1}{k^2}-\frac{|m|^2}{k^3}, \ldots, \frac{1}{k^2}-\frac{|m|^2}{k^3}}_{k^2-1 \text{ times}} \right),\\
\lambda^{RUC}(Z_n) &\to \left( \frac{|m|^2}{k} + \frac{1-|m|^2}{k^2}, \underbrace{\frac{1-|m|^2}{k^2}, \ldots, \frac{1-|m|^2}{k^2}}_{k-1 \text{ times}} , \underbrace{\frac{1}{k^2}, \ldots, \frac{1}{k^2}}_{k^2-k \text{ times}} \right).
\end{align*}

Finally, let us compare output spectra in the regime $k/n \to c$. We shall make the same assumption $t=1/k$ for the results of \cite[Theorem 6.11]{cn3} and also we shall consider usual Bell states, which imposes $|m|^2=1$ for Theorem \ref{thm:k-sim-n}. In the case where the coefficients $w$ are ``flat'', i.e.\ $\nu=1$ in Theorem \ref{thm:k-sim-n}, the results are \emph{identical} for random channels and random unitary channels: the largest eigenvalue behaves like $1/(cn)$ and the lower spectrum has a limiting shape $\pi_{c^2}$ on the scale $1/n^2$. 

In order to understand whether different weights for the unitary operators in \eqref{eq:def-Phi} are more interesting for the purpose of finding counterexamples to additivity relations, one needs to understand how the minimal output entropy behaves for a single copy of such channels - this is the subject of future work \cite{cfn2}.

\section*{Acknowledgments}

The three authors would like to thank the organizers of the workshop ``Probabilistic Methods in Quantum Mechanics'' in Lyon, where the project was initiated, and the ANR project HAM-MARK. B.~C.~'s research was supported by an NSERC Discovery grant and an ERA at the University of Ottawa. I.~N.~ acknowledges financial support from a PEPS grant from the CNRS, the AO1-SdM project ECIU and the ANR project OSvsQPI 2011 BS01 008 01.

\end{document}